\documentclass[aps,prb,twocolumn,superscriptaddress]{revtex4-1}
\usepackage{H1}
\usepackage[dvipsnames]{xcolor}

\renewcommand\textemdash{\leavevmode\unskip\kern0.8pt\rule[0.19\baselineskip]{8pt}{0.4pt}\kern1pt\ignorespaces}

\makeatletter

\newtheorem{corollary}{Corollary}[theorem]
\newtheorem{fact}{Fact}[section]

\begin{document}

\title{Operator hydrodynamics, OTOCs, and entanglement growth in systems without conservation laws}
\author{C.W.~von~Keyserlingk}
\thanks{These authors made roughly equal contributions to this work.}
\affiliation{Department of Physics, Princeton University, Princeton, New Jersey 08544, USA}
\author{Tibor Rakovszky}
\thanks{These authors made roughly equal contributions to this work.}
\affiliation{Technische Universit{\"a}t M{\"u}nchen, Physics Department T42, 85747 Garching, Germany}
\author{Frank Pollmann}
\affiliation{Technische Universit{\"a}t M{\"u}nchen, Physics Department T42, 85747 Garching, Germany}
\author{S.~L.~Sondhi}
\affiliation{Department of Physics, Princeton University, Princeton, New Jersey 08544, USA}
\begin{abstract}
Thermalization and scrambling are the subject of much recent study from
the perspective of many-body quantum systems with locally bounded Hilbert spaces (``spin
chains''), quantum field theory and holography. We  tackle this problem in 1D spin-chains evolving under random local unitary circuits and prove a number of exact results on the behavior of out-of-time-ordered commutators (OTOCs), and entanglement growth in this setting. These results follow from the observation that the spreading of operators in random circuits is described by a ``hydrodynamical'' equation of motion, despite the fact that random unitary circuits do not have locally conserved quantities (e.g., no conserved energy). In this hydrodynamic picture quantum information travels in a front with a `butterfly velocity' $v_{\text{B}}$ that is smaller than the light cone velocity of the system, while the front itself broadens diffusively in time. The OTOC increases sharply after the arrival of the light cone, but we do \emph{not} observe a prolonged exponential regime of the form $\sim e^{\lambda_\text{L}(t-x/v)}$ for a fixed Lyapunov exponent $\lambda_\text{L}$. We find that the diffusive broadening of the front has important consequences for entanglement growth, leading to an entanglement velocity that can be significantly smaller than the butterfly velocity. We conjecture that the hydrodynamical description applies to more generic ergodic systems and support this by verifying numerically that the diffusive broadening of the operator wavefront also holds in a more traditional non-random Floquet spin-chain. We also compare our results to Clifford circuits, which have less rich hydrodynamics and consequently trivial OTOC behavior, but which can nevertheless exhibit linear entanglement growth and thermalization.
\end{abstract}
\maketitle
\section{Introduction}\label{s:intro}
The past decade has seen a great revival of interest in the foundations of quantum statistical mechanics.
It has been driven by theoretical advances involving the long sought demonstration that many-body
localization (MBL)  exists~\cite{Basko06}, ideas from quantum information theory and the study of integrable systems~\cite{Essler16}. It has
been equally driven by experimental advances in the study of cold atomic gases which provide examples
{\it par excellence} of closed macroscopic quantum systems for which the foundational questions of
quantum statistical mechanics are especially acute.~\cite{Bloch2008}  Perhaps the broadest question has to do with 
identifying possible ``ergodic universality classes'' in quantum many-body systems and understanding
their more detailed physics. Of such a potential classification much work has focused on fully MBL systems
which are believed to exhibits a breakdown of statistical mechanics.~\cite{Nandkishore14}

Most recently a great deal of attention has focused on the related question of how to quantify ``scrambling'' in many-body systems\cite{Sekino08,Brown12,Lashkari2013,Shenker2014b,Hosur2016}. In this work, we will use the word scrambling to denote those features of the spreading of quantum information which are quantified by the out-of-time-ordered commutator (OTOC), which has been studied in the SYK model and its descendants\cite{Sachdev93,Kitaev15}, as well as MBL systems\cite{MBLOTOC1,MBLOTOC2,MBLOTOC3,MBLOTOC4,MBLOTOC5}, in field theoretic settings\cite{Stanford15,Stanford2016,Swingle17,Aleiner16,Dubail16}, and numerically in interacting spin-chains\cite{Moessner16,Bohrdt16,Prosen17,Cheryne}. We report some exact analytical results and supporting numerics for \textit{interacting non-integrable} spin-chains that are interesting in the context of scrambling.

In the following we discuss OTOC behavior and entanglement growth (building on work in Refs.~\onlinecite{Nahum16,Abanin17,Mezei16}) for the following three spin-chain models: I) a random circuit (see Eq.~\eqref{eq:circuit_def}) where the 2-site gates are randomly chosen; II) an ergodic Floquet system with nearest-neighbor interactions, defined in Eq.~\eqref{eq:Ising_def} and III) a periodic Clifford circuit defined in Eq.~\eqref{eq:fractalClifford}. Our approach relies on quantifying operator spreading, i.e., how the support of operators changes under Heisenberg picture evolution. We derive analytical formulas for operator spreading in model I), which we support with additional numerics, while for model II) we rely entirely on numerical calculations. Our numerical method is based on the matrix product operator (MPO)~\cite{VerstraeteMPO} representation. Since all three types of time evolution we consider can be represented as a network of 2-site gates (see Fig.~\ref{fig:circuit}), the MPO can be time evolved straightforwardly by using the TEBD algorithm~\cite{VidalTEBD}. Our results for the three models are as follows:

I) In the random circuit model (\secref{s:randomcircuit}) we find that operator spreading can be described by a remarkably simple hydrodynamical picture, which gives rise to a biased diffusion equation. Using this, we find that the typical extent of an operator grows with butterfly velocity $v_{\text{B}}$ which is less than the light cone velocity $v_{\text{LC}}$, while the width of the front broadens diffusively in time (see \figref{fig:front_propagation}). We use these results to derive exact formulas for the OTOC and entanglement growth. The OTOC travels with the same butterfly velocity $v_{\text{B}}$ and its behavior near the front is also sensitive to the diffusive broadening (\eqnref{eq:otoc_intermediate}). At early times, before the arrival of the main front, the OTOC grows exponentially, with an exponent that increases with the initial separation of the two operators (\eqnref{eq:lyaponov}). At long times the OTOC saturates to $1$. This is summarized in Fig.~\ref{fig:random_circuit_otoc}. Our front propagation results also lead to an exact formula for the entanglement growth of an initial product state, from which we can extract an entanglement velocity $v_\text{E}$ (Eq.~\eqref{eq:entanglement_velocity}). We find that the diffusive broadening of the operator front gives rise to the inequality $v_\text{E}<v_{\text{B}}$. This exact result is consistent with general non-rigorous arguments\cite{Jeddi15,Mezei16}, the heuristic operator spreading model of Ref. \onlinecite{Mezei16}, numerous results in holography\cite{Hartman2013,Liu14a,Liu14b}, and the results derived for Clifford circuits in Ref. \onlinecite{Nahum16}. 

II) In \secref{s:kickedIsing} we verify numerically that, for a family of ergodic Floquet circuits, there is a similar diffusively broadening front behavior as observed in the random circuit (see Fig.~\ref{fig:KimHuse_front_widening}). This leads to the tentative conjecture that the diffusive front picture is valid for generic ergodic 1D spin-chains, along with the resulting consequences for OTOC and entanglement dynamics.

III) Finally in \secref{s:Clifford} we compare I) and II) to Clifford circuits. Within such circuits, strings of Pauli operators evolve to other particular strings, rather than superpositions of such; in particular, Clifford circuits do not exhibit a diffusively broadening operator front. We connect this fine tuned nature of local Clifford circuits to the fact that such circuits always exhibit trivial OTOC behavior. Nevertheless, they can still have linear entanglement growth and their local  observables thermalize to infinite temperature. This demonstrates the broader point that the presence of ballistic entanglement growth and thermalization are not sufficient to predict scrambling behavior (i.e., the behavior of the OTOC).

The systems above do not have local conserved quantities (in particular, they do not conserve energy), so it is surprising that hydrodynamics arises in I) and II).  We conjecture that such hydrodynamic behavior is universal, generically appearing in 1D ergodic systems with \textit{local unitary} evolution and a bounded local Hilbert space. [In this connection, our work has obvious parallels with Ref.~\onlinecite{Nahum16}, although as we will discuss in \secref{s:randomcircuit} the hydrodynamics in our case has a different physical origin]. Moreover, our work makes a clear and precise connection between the spreading dynamics of operators, the scrambling behavior captured by the OTOC and other metrics of ergodicity such as entanglement entropy and the late time behavior of local correlation functions (see \appref{ss:longtimerandom}).

\section{Quantifying operator spreading}\label{s:QuantifyOperatorSpread}

Consider a one-dimensional chain of $L$ sites, for which the Hilbert space of a single site is $\mathcal{H}_\text{site}=\mathbb{C}^{q}$.
There exist operators $X,Z$ on the single site Hilbert space obeying
\begin{subequations}
\begin{align}
ZX & = e^{2\pi i/q}XZ\\
Z^{q} & =X^{q} =1\punc{.}
\end{align}
\end{subequations}
These generate a convenient complete basis for all operators on $\mathcal{H}_\text{site}$, namely $\{ \sigma^{\mu}\equiv X^{\mu^{(1)}}Z^{\mu^{(2)}}:\mu\in\mathbb{Z}_{q}^{\otimes2}\} $. Here $\mu$ is shorthand for the doublet $\mu^{(1)},\mu^{(2)}\in\left\{ 0,1,\ldots,q-1\right\} =\mathbb{Z}_{q}$. This basis is orthonormal, such that $\mathrm{tr}(\sigma^{\mu\dagger}\sigma^\nu) / q = \delta_{\mu\nu}$. The operators $\sigma^\mu$ can be regarded as generalizations of Pauli matrices, where the usual Paulis correspond to the $q=2$ case. Generalizing this to the Hilbert space of a 1D chain, $\mathcal{H}_\text{chain}=\left(\mathbb{C}^{q}\right)^{\otimes L}$, a complete orthonormal basis of operators is given by the $q^{2L}$ Pauli \emph{strings}, defined as $\sigma^{\vec{\mu}}\equiv\bigotimes_{r=1}^{L}\sigma_{r}^{\mu_{r}}$, where each string is indexed by a vector $\vec{\mu}\in\left(\mathbb{Z}_{q}^{\otimes2}\right)^{\otimes L}$.

Our goal is to quantify how an initial Pauli string spreads over the space of all Pauli strings under local unitary time evolution. At time $\tau$ the Pauli string $\sigma^{\vec{\mu}}$ becomes
\be\label{eq:cmunudef}
\sigma^{\vec{\mu}}\left(\tau \right)  \equiv U^{\dagger}(\tau)\sigma^{\vec{\mu}}U(\tau)  =\sum_{\nu}c_{\vec{\nu}}^{\vec{\mu}}\left(\tau\right)\sigma^{\vec{\nu}}.
\ee
This defines a set of `operator spread coefficients' $c_{\vec{\nu}}^{\vec{\mu}}\left(\tau\right)\equiv\text{tr}\left(\sigma^{\vec{\nu}\dagger}U^{\dagger}(\tau)\sigma^{\vec{\mu}}U(\tau)\right) / q^L$. The full set of coefficients $\{c_{\vec{\nu}}^{\vec{\mu}}(\tau)\}$ encodes all information regarding the unitary time evolution. However, as we show below, for accessing most physically interesting quantities, such as entanglement entropies\cite{Abanin17,Mezei16,Nahum16}, or out-of-time-order commutators, it is sufficient to consider a more coarse-grained description of the operator spreading. One particularly useful coarse-grained quantity is the total weight on all operators with right endpoint $s$ appearing in $\sigma^\mu(\tau)$, i.e.,
\begin{equation}\label{eq:rho_def}
\rho_R^{\vec{\mu}}(s,\tau) \equiv \sum_{\vec{\nu}} \left|c_{\vec{\nu}}^{\vec{\mu}}\left(\tau\right)\right|^{2} \delta(\text{RHS}(\vec{\nu}) = s),
\end{equation}
where $\text{RHS}(\vec{\nu})$ denotes the rightmost site on which $\vec{\nu}$ is non-zero.~\footnote{A similar quantity was defined in Ref. \onlinecite{Roberts2015}} Note that the `density' $\rho^{\vec{\mu}}_R$ is conserved, i.e., $\sum_s \rho^{\vec{\mu}}_R(s) = 1$ at all times. Motivated by this, we refer to $\rho_R^{\vec{\mu}}$ as the \emph{operator density} of the time evolved Pauli string $\sigma^{\vec{\mu}}$.

In this paper we consider systems where the time evolution can be represented as a circuit of 2-site unitary gates, arranged in the geometry shown in~\figref{fig:circuit}. Sites of the 1D chain are indexed by $s=1,\ldots,L$ while the layers of the circuit are indexed by the variable $\tau$. It is useful to introduce coarse-grained coordinates $x$ and $t$ that label pairs of sites and pairs of layers, respectively, as defined in Eq.~\eqref{eq:newrho} and illustrated in~\figref{fig:circuit}. Note that due to the geometry of the circuit all such models have a well defined light cone velocity $v_\text{LC} = \Delta s/\Delta\tau = \Delta x/\Delta t = 1$, corresponding to the fact that with each successive time-step a local operator can spread at most one additional site in each direction.

In the following we investigate the three models I)-III) discussed in the introduction, all three of which can be represented by local circuits of the kind shown in Fig.~\ref{fig:circuit}. In Sec.~\ref{ss:randomwalkdynamics} we show that for model I) the average of $\rho^{\vec{\mu}}_R$ obeys a classical biased diffusion equation, and use this to derive exact formulas for the behavior of OTOCs and the dynamics of entanglement. In Sec.~\ref{s:kickedIsing} we investigate model II) and find that it shares many features with the random circuit model, such as a broadening of the propagating wavefront, which suggests that the random walk description obtained in Sec.~\ref{s:randomcircuit} has applications in a wider class of ergodic systems. In Sec.~\ref{s:Clifford} we contrast this with model III) where there is no diffusion and $\rho_R^\mu$ remains a delta-function at all times, which corresponds to a non-generic behavior of the OTOC.

 \begin{figure}[h!]
 \centering
  	\includegraphics[width=0.3\textwidth]{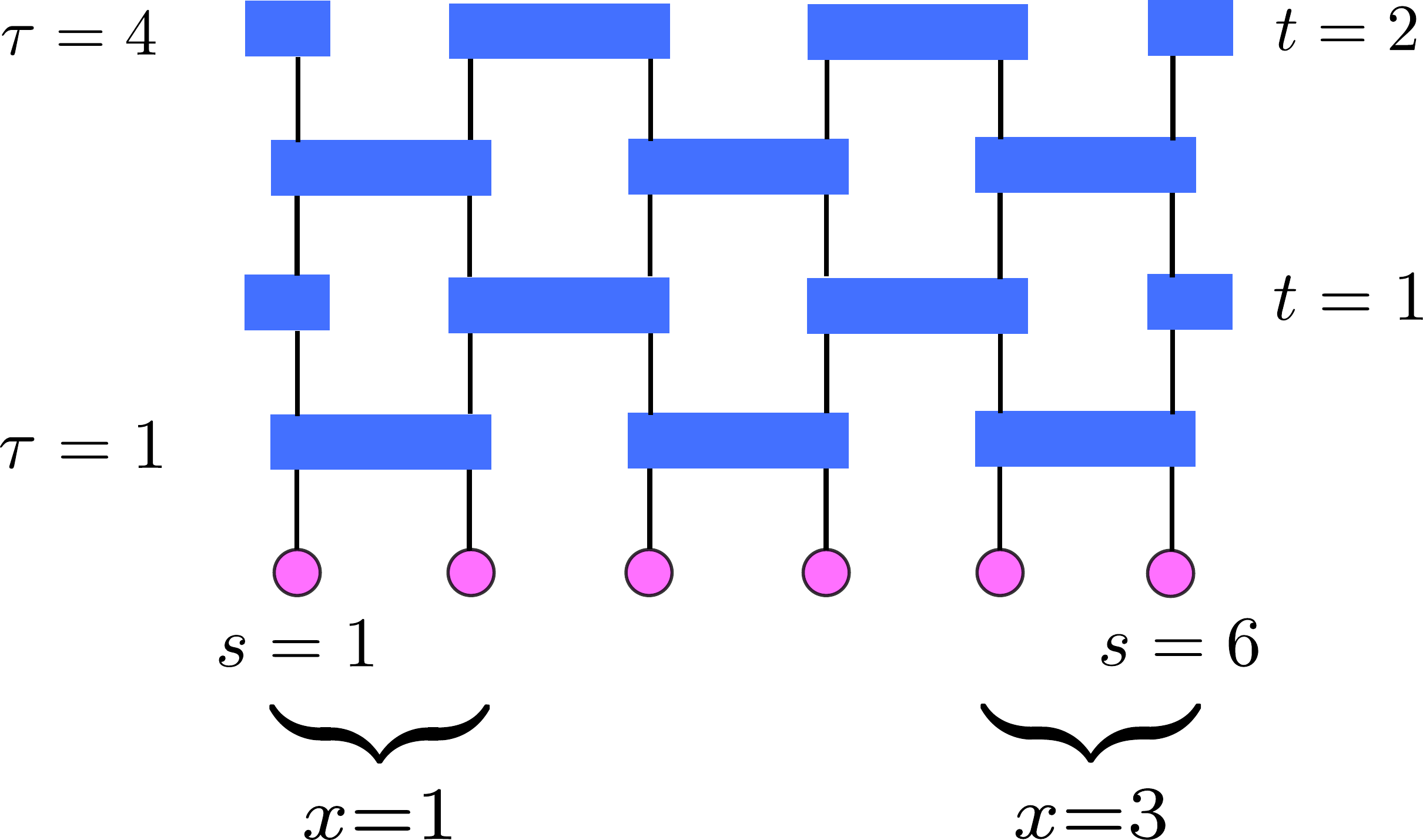} 
\caption{Structure of the local unitary circuits studied in this paper. The on-site Hilbert space dimension is $q$. Each two-site gate is a $q^2 \times q^2$ unitary matrix. For the random circuit model of Eq.~\eqref{eq:circuit_def} each gate is randomly chosen from the Haar distribution. For the Floquet models considered in Sec.~\ref{s:kickedIsing} and~\ref{s:Clifford} the 2-site gates are defined by the Floquet unitaries in Eqs.~\eqref{eq:Ising_def} and~\eqref{eq:fractalClifford}, respectively.}
 \label{fig:circuit}
 \end{figure}

\section{Random circuit model}\label{s:randomcircuit}
Most of the present paper is concerned with one dimensional local random unitary circuits, with the geometry shown in Fig.~\ref{fig:circuit}. Random circuits were also investigated in Ref.~\onlinecite{Nahum16} with regards to the growth of entanglement from an initial product state, albeit with a different geometry where the location of the unitary gates is randomly chosen, instead of the regular arrangement used here. There it was argued that the evolution of entanglement obeys an equation belonging to the KPZ universality class, which determines certain universal exponents that appear in the average value and fluctuations of the entanglement entropy. Here we shift our focus from states to operators and derive exact results for their spreading, for arbitrary on-site Hilbert space dimension. In~\secref{ss:randomcircuit_entanglement} we relate our operator spreading results to the dynamics of bipartite entanglement, as captured by the second R\'enyi entropy, and find no sign of the universal fluctuations observed in Ref.~\onlinecite{Nahum16}. This suggests that the KPZ-like behavior is specific to circuits whose geometry (and not just the individual circuit elements) is random (compare \figref{fig:circuit} above to Fig.~6 of Ref.~\onlinecite{Nahum16}).

The random circuits we discuss are defined as follows. Consider a discrete time evolution, consisting of layers of two-site unitary gates acting on pairs of neighboring sites. Odd numbered layers act on all the odd bonds of the chain while even numbered layers act on even bonds. Each two-site gate is chosen independently from the Haar distribution over $q^2\times q^2$ unitary matrices. The time evolution after an even number of $2t$ layers is given by
\begin{align}\label{eq:circuit_def}
U(t)=\prod_{\tau=1}^{2t,\leftarrow}\prod_{x=1}^{L/2}W(2 x-1+n_{\tau},\tau)
\end{align}
where $n_{\tau}=\frac{1+(-1)^{\tau}}{2}$ and $W(s,\tau)$
is a Haar random two site unitary acting on sites $s,\,s+1$. The product $\prod_{\tau=1}^{2t,\leftarrow}$ is defined to be time ordered. Such a circuit is graphically illustrated in Fig.~\ref{fig:circuit}. 

The primary goal of this work is to quantify the spread of operators under random circuits \eqnref{eq:circuit_def}, and to relate operator spread to entanglement growth. A related question is how correlation functions of local observables behave in this random circuit model in the thermodynamic limit. As we confirm in~\appref{ss:longtimerandom}, such correlations tend to their infinite temperature values at long times, similarly to the case of Floquet ergodic systems\cite{Lazarides14PRL, Lazarides14PRE,Rigol14,Abanin14,Ponte15}. This result holds for any random realization of the circuit.

Focusing on the problem of operator spreading in this random circuit model, we find that the average of the operator density $\rho_R$, defined in Eq.~\eqref{eq:rho_def}, performs a biased random walk, independent of the internal structure of the operator considered. Solving the random walk problem allows us to derive exact formulas for the averages of out-of-time-order commutators and entanglement growth, which we detail in Sec.~\ref{ss:randomcircuitOTOC} and~\ref{ss:randomcircuit_entanglement}, respectively. In \secref{ss:fluct} we quantify numerically the fluctuations between different random realizations of the circuit.

\subsection{Random walk dynamics of operator density}\label{ss:randomwalkdynamics}

In the following we will quantify how operators spread under the time evolution generated by the random circuit defined above in Eq.~\eqref{eq:circuit_def}. We focus on the average of the operator density, defined in Eq.~\eqref{eq:rho_def}, for which we derive an exact equation of motion. Upon solving this equation we find that the operator density moves in a front whose velocity $v_{\text{B}}$ is an increasing function of the on-site Hilbert space dimension and with a front width increasing diffusively in time. 

We start by noting that under Haar averaging the operator spread coefficient $c_{\vec{\nu}}^{\vec{\mu}}(\tau)$ vanishes for any time $\tau\geq1$, provided that $\vec{\mu}$ is non-trivial, since $c_{\vec{\nu}}^{\vec{\mu}}$ and $-c_{\vec{\nu}}^{\vec{\mu}}$ have equal probability. However, the average of its modulus squared, $\overline{|c_{\vec{\nu}}^{\vec{\mu}}(\tau)|^{2}}$, can be non-zero. (An explicit expression for this quantity is written down in \appref{App:cmunu}, using a mapping to a classical Ising model, but we will not require it for the subsequent discussion).

Following \eqnref{eq:rho_def}, we define the average operator density as
\begin{equation}
\overline{\rho_R^{\vec{\mu}}}(s,\tau) \equiv \sum_{\vec{\nu}} \overline{|c_{\vec{\nu}}^{\vec{\mu}}(\tau)|^{2}}\,\delta(\text{RHS}(\vec{\nu}) = s).
\end{equation}
As we show below, this quantity satisfies an equation of motion, \eqnref{eq:random_walk}, which does not depend explicitly on $\vec{\mu}$. Pre-empting this, we drop the explicit $\vec{\mu}$ dependence $\overline{\rho_R^{\vec{\mu}}}\rightarrow \overline{\rho_R}$ to declutter notation. In fact, $\vec{\mu}$ will enter considerations only as an initial condition on the operator density
\begin{equation}\label{eq:rho_initial}
\overline{\rho_R}(s,0) = \delta(\text{RHS}(\vec{\mu}) = s),
\end{equation}
which is the same for all initial operators sharing the same right endpoint.

To understand how $\overline{\rho_R}$ evolves in time, consider the effect of applying a single two-site gate on sites $s$ and $s+1$. There are $q^4-1$ nontrivial operators acting on this two-site Hilbert space. Of these, $q^2-1$ contribute to $\overline{\rho_R}(s,\tau)$ (the ones that are trivial on site $s+1$), while the other $q^2(q^2-1)$ contribute to $\overline{\rho_R}(s+1,\tau)$. Under a two site Haar random unitary transformation all the possible transitions between any of these $q^4-1$ operators have, on average, the same probability\cite{Brown12}. The upshot is that after the application of the unitary gate the density $\overline{\rho_R}$ evolves as
\begin{subequations}\label{eq:hopping_one_gate}
\begin{align}
&\overline{\rho_R}(s,\tau+1) = (1-p) \left[\overline{\rho_R}(s,\tau) + \overline{\rho_R}(s+1,\tau) \right]; \\
&\overline{\rho_R}(s+1,\tau+1) = p \left[\overline{\rho_R}(s,\tau) + \overline{\rho_R}(s+1,\tau) \right],
\end{align}
\end{subequations}
with probabilities $p = \frac{q^2}{q^2+1}$ and $1-p = \frac{1}{q^2+1}$. To apply a similar argument for two subsequent layers of the circuit it is useful to redefine the density by grouping together the pairs of sites on which the first layer of the circuit acts. We abuse notation and denote this quantity as
\begin{equation}\label{eq:newrho}
\overline{\rho_R}(x,t) \equiv  \overline{\rho_R}(s=2x-1,\tau=2t) + \overline{\rho_R}(s=2x,\tau=2t),
\end{equation}
where we now only consider the value of the operator density at even time steps $\tau = 2t$. Applying Eq.~\eqref{eq:hopping_one_gate} for two layers we arrive at the equation 
\begin{multline}\label{eq:random_walk}
\overline{\rho_R}(x,t+1) = 2p(1-p)\,\overline{\rho_R}(x,t) + \\
+ p^2\,\overline{\rho_R}(x-1,t) + (1-p)^2\,\overline{\rho_R}(x+1,t).
\end{multline}
Thus the right endpoints of Pauli strings perform a biased random walk on the lattice, where in each step they move to the right with probability $p^2$, to the left with probability $(1-p)^2$, and stay on the same site otherwise. A feature of the above equation is that the time evolution of $\overline{\rho_R}$ is independent of the internal structure of the operator and thus the solution $\overline{\rho_R}(x,t)$ will be the same for all initial Pauli strings, modulo a shift $x\to x-x_0$ where $x_0$ is defined by the right endpoint of the initial string.

The result of the random walk process outlined above is a front that propagates to the right from its initial position $x_0$ as $\langle x\rangle - x_0 = v_{\text{B}} t$ with a butterfly velocity $v_{\text{B}} = p^2 - (1-p)^2 = \frac{q^2-1}{q^2+1}$. Thus the speed at which the operator weight travels is smaller than the light cone velocity: $v_{\text{B}} < v_\text{LC} = 1$. This resonates somewhat with the result of Ref.~\onlinecite{Roberts16}. The width of the front increases in time as $\langle x^2 \rangle - \langle x\rangle^2 = 2Dt$ with diffusion constant $D = \sqrt{1-v_{\text{B}}^2} / 4 = \frac{q/2}{q^2 + 1}$. Note that in the limit $q\to\infty$ the `particle' described by $\overline{\rho_R}(x,t)$ hops to the right with probability $1$ in each step, and consequently the front becomes infinitely sharp with velocity $v_{\text{B}}\to v_\text{LC}=1$.

The total weight of left endpoints, $\overline{\rho_L}(x,t)$, obeys a similar equation except that it propagates to the left with velocity $-v_{\text{B}}$, while diffusing at the same rate, as shown in Fig.~\ref{fig:front_propagation}. This means that at time $t$ the vast majority of quantum information initially stored in $\sigma^{\vec{\mu}}$ with left (right) endpoint $x_l$ ($x_r$) is carried by operators with support $[x_l-v_{\text{B}}t,x_r+v_{\text{B}}\tau]$, but where the precise position of either endpoint can be uncertain within a region of width $\Delta x \sim \sqrt{D t}$.  
 \begin{figure}[h!]
 \centering
  	\includegraphics[width=0.485\textwidth]{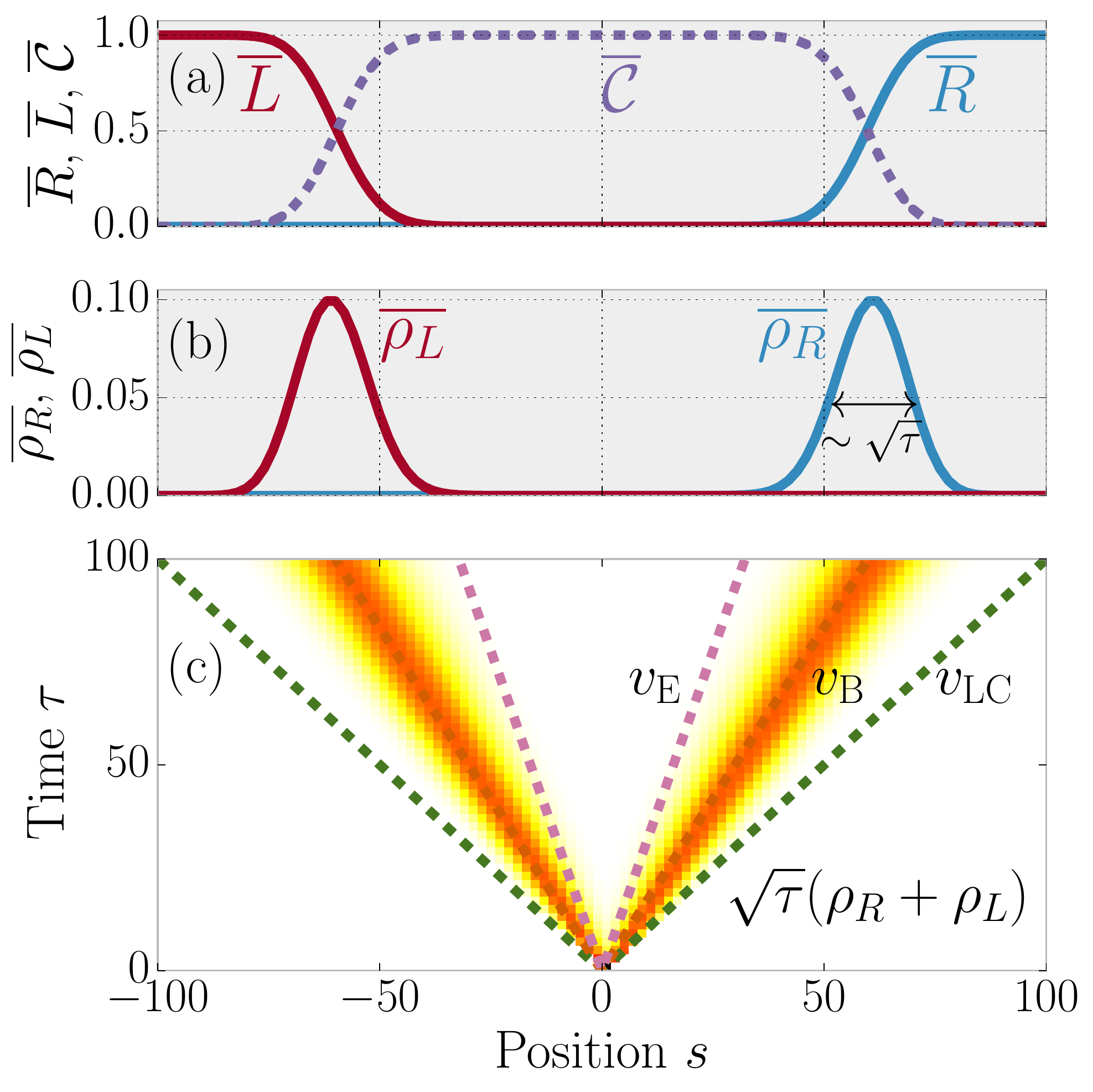} 
\caption{Spreading of a one-site operator averaged over random unitary circuits. $\overline{\rho_R}(s,\tau)$ ($\overline{\rho_L}(s,\tau)$) is the total weight carried by Pauli strings with right (left) endpoint at site $s$ at time $\tau$. Figure (c) shows the sum of these two functions (multiplied by $\sqrt{\tau}$ to show the position of the front more clearly). Almost all the weight is carried by operators with endpoints at the two fronts propagating out from the initial site with speed $v_{\text{B}} = \frac{q^2-1}{q^2+1}$. These fronts in turn broaden diffusively in time as $~\sqrt{\tau}$. The two other velocity scales, the light cone velocity $v_\text{LC}$ and the entanglement velocity $v_\text{E}$ (see Eq.~\eqref{eq:entanglement_velocity}) are also indicated, satisfying $v_\text{E} < v_{\text{B}} < v_\text{LC}$. The values of $\overline{\rho_R}$ and $\overline{\rho_L}$ after 100 layers of the circuit are shown in Fig. (b). Fig. (a) shows the integrated operator weights $\overline{R}(s)$ ($\overline{L}(s)$), denoting the total weight left (right) of site $s$, along with the OTO commutator $\mathcal{C}(s,\tau)$. The OTOC saturates to 1 inside the front and has the value $1/2$ exactly at $\tau = s/v_{\text{B}}$}
 \label{fig:front_propagation}
 \end{figure}

We can find the full distribution of $\overline{\rho_R}(x,t)$ using a standard generating functional method. In the rest of this section we will use coordinates relative to the initial position of the front, i.e. $x-x_0 \to x$. The solution to Eq.~\eqref{eq:random_walk} than reads
\begin{equation}\label{eq:rho_solution}
\overline{\rho_{R}}(x,t) = \frac{q^{2(t+x)}}{(1+q^2)^{2t}} {2t \choose t+x}.
\end{equation}
In the scaling limit $t,x\to\infty$ but keeping $x/t \approx v_{\text{B}}$ fixed this becomes (using Stirling's approximation)
\begin{equation}
\overline{\rho_{R}}(x=v_{\text{B}}t+O(\sqrt{t}))=\frac{1}{\sqrt{\pi(1-v_{\text{B}}^{2})t}}e^{-\frac{(x-v_{\text{B}}t)^{2}}{(1-v_{\text{B}}^{2})t}}\label{eq:rhonearfront},
\end{equation}
so that the traveling front has the shape of a Gaussian, as one would expect from the solution of the continuum limit of the lattice diffusion equation Eq.~\eqref{eq:random_walk}.

As we shall see in the next section, it is also useful to compute the total weight of all Pauli strings contained entirely to the left of position $x$. This quantity, which we denote by $R(x),$ is given by
\begin{align}
R(x) & \equiv\sum_{y\leq x}\rho_{R}(y).
\label{eq:Rdefn}
\end{align}
Around the position of the front, where $x\approx v_{\text{B}}t$, we can integrate Eq. \eqref{eq:rhonearfront}
to obtain 
\begin{equation}
\overline{R}(x=v_{\text{B}}t+O(\sqrt{t}))\approx\frac{1}{2}\left[\text{erf}(\frac{x-t v_{\text{B}}}{\sqrt{t(1-v_{\text{B}}^{2})}})+1\right],\label{eq:frontRscaling}
\end{equation}
where $\text{erf}(x)$ is the error function.

Later on we will also need an approximation for $R$ well away from the front. Using
the fact that in the large $t$ limit $\overline{\rho_{R}}(x,t)$ increases sharply with $x$ for $x/t<v_{\text{B}}$, the sum in Eq.~\eqref{eq:Rdefn} is dominated by its largest term
(i.e. $y=x$). Using the fact that $\overline{R}\left(x\geq v_\text{LC} t \right)=1$ we can similarly approximate $\overline{R}$ for $x/t>v_{\text{B}}$, to obtain
\begin{equation}\label{eq:Rapprox}
\overline{R}(x)\approx(1-2\overline{\rho_{R}}(x))\Theta(x-v_{\text{B}}t)+\overline{\rho_{R}}(x),
\end{equation}
where $\Theta$ is the Heaviside step function. This result is accurate up to multiplicative $O(1)$ constants when $|x-v_{\text{B}}t|/t = O(1)$ in the large $t$ limit. See the discussion in \appref{app:approx} for a precise statement and derivation of \eqnref{eq:Rapprox}. 

Using our results for the coarse-grained density $\overline{\rho_R}(x,t)$ we can also write a formula for the density in terms of the site coordinate $s$. Note that due to Eq.~\eqref{eq:hopping_one_gate} the ratio $\overline{\rho_R}(s=2x+1) / \overline{\rho_R}(s=2x) = q^2$ is fixed at any time $\tau = 2t$. Using this, the density on site $s$ after applying an even number of layers becomes
\begin{equation}\label{eq:rho_realspace}
\overline{\rho_{R}}(s=2x+n,\tau = 2t) = \frac{q^{2(t+x-1+n)}}{(1+q^2)^{2t}} {2t-1 \choose t+x-1},
\end{equation}
where $n=0,1$. We can use Eq.~\eqref{eq:rho_realspace} to derive the total operator weight left of site $s$, i.e. $R(s) = \sum_{r\leq s}\rho(r)$, which, as we will see in the next section, is closely related to the OTOC between sites $0$ and $s$.

\subsection{Behavior of out-of-time-order commutators}\label{ss:randomcircuitOTOC}
We relate our results for the time evolution of operator weights to another oft used measure of information spreading in many-body systems, the so-called out-of-time-order commutator (OTOC)\cite{Shenker2014a,Shenker2014b,Shenker2015,Roberts2015,Stanford15,Maldacena2016,Hosur2016}. For concreteness consider the following OTOC between two Pauli operators separated by distance $s$ (in this section we work in a shifted co-ordinate system where one of the Pauli operators resides at site $0$)
\begin{align}\label{eq:OTOC_def}
\mathcal{C}(s,\tau) & \equiv  \frac{1}{2}\langle\psi\mid\left|\left[Z_{0}\left(\tau \right),Z_{s}\right]\right|^{2}\mid\psi\rangle \nonumber \\
 & =  1-\text{Re}\langle\psi\mid Z_{0}\left(\tau \right)Z_{s}Z_{0}^{-1}\left(\tau \right)Z_{s}^{-1}\mid\psi\rangle
\end{align}
where $s$ and $\tau$ are the original time/lattice co-ordinates (as opposed to the coarse grained co-ordinates $t,x$ below \eqnref{eq:newrho}). We will show how the OTOC $\mathcal{C}(s,\tau) $ behaves in the scaling limit $\tau\rightarrow \infty$, with $\kappa \equiv s/\tau$ held fixed. We we detail below, for $s$ outside of the light cone ($1<\kappa$) it is zero. As $s$ enters the light-cone ($\kappa<1$  and close to $1$) it increases exponentially.   When $s$ is near the operator front ($\kappa = v_{\text{B}}<1$) the OTOC becomes $O(1)$. After the front has passed ($\kappa < v_{\text{B}}$) the front has passed, the OTOC exponentially saturates to the value $1$ with an exponent that is independent of $s$. See \figref{fig:random_circuit_otoc} for a summary. 

Let $\overline{\mathcal{C}}$ denote the average of the OTOC over all unitary circuits with geometry shown in \figref{fig:circuit}. Note that due to the averaging this quantity is independent of the choice of Pauli operator, i.e., it would be the same if we replaced either or both operators in \eqnref{eq:OTOC_def} with another local Pauli different from $Z$. We will be concerned with the second term in Eq.~\eqref{eq:OTOC_def} which equals

\begin{eqnarray}\label{eq:OTOC_from_c}
1-\overline{\mathcal{C}}(s,\tau) & = & \sum_{\mu\nu}\overline{c_0^{\mu}(\tau)c_0^{*\nu}(\tau)}\,\text{Re}\langle\psi\mid\sigma^{\mu}Z_{s}\sigma^{\dagger\nu}Z_{s}^{-1}\mid\psi\rangle \nonumber\\
 & = & \sum_{\mu}\overline{\left|c_0^{\mu}(\tau)\right|^{2}}\cos\theta_{\mu,Z_{s}},
\end{eqnarray}
where $e^{\theta_{\mu,Z_{s}}}$ is  a $q^{\text{th}}$ root of unity arising from commuting $\sigma^{\mu}$ past $Z_{s}$, and $c_0^{\mu}(\tau)$ are the operator spreading coefficients of $Z_0(\tau)$. Notice that the Haar average forces $\mu=\nu$ which removes all dependence on the particular initial state $\psi$\footnote{To see that $\overline{c^\mu c^{\nu*}(\tau)} \propto \delta_{\mu \nu}$: If $\mu\neq\nu$ there exists a single site Pauli $\sigma^\alpha_r$ such that the group commutators $[\sigma^{\alpha}_r:\sigma^{\mu,\nu}]=e^{i \theta_{\mu,\nu}}$ are unequal. Using the invariance of the two-site Haar measure under multiplication by single site operators, this implies $\overline{c^\mu c^{\nu*}(\tau)} = e^{i(\theta_{\mu} - \theta_{\nu})} \overline{c^\mu c^{\nu*}(\tau)}$ from whence the result follows.}. In particular, the average OTOC value in any state, pure or mixed, will be identical to the average OTOC value at infinite temperature, i.e., $\text{tr} (\frac{1}{2}|[Z_{0}(\tau),Z_{s}]|^{2})/2^L$. At this point, we can use \eqnref{eq:OTOC_from_c} and \eqnref{eq:Cmunufull} to write an exact closed from expression for the OTOC. However, instead of doing that, we will write a more manageable asymptotic expression for \eqnref{eq:OTOC_from_c} using simpler results from \secref{ss:randomwalkdynamics}.

To perform the sum over Pauli strings in Eq.~\eqref{eq:OTOC_from_c}, we first need to prove the following statement: $\overline{|c_0^{\vec{\mu}}(\tau)|^2}$ depends \emph{only on the position of the two endpoints of the string} $\vec{\mu}$. The proof goes as follows. First, it is easy to verify that under Haar averaging $|c_{\vec{\nu}}^{\vec{\mu}}(\tau)|^2 = |c_{\vec{\mu}}^{\vec{\nu}}(\tau)|^2$, which means that the average probability of the one-site operator $Z_0$ evolving into a specific string $\vec{\mu}$ is the same as the probability of string $\vec{\mu}$ evolving into $Z_0$. In the random walk picture this latter process corresponds to both left and right endpoints ending up on site $0$ at time $\tau$ during their respective random walks. As we noted previously, these random walks are independent of the internal structure of the initial string. Thus $\overline{|c_0^{\vec{\mu}}(\tau)|^2}$ depends only on where the two endpoints of $\vec{\mu}$ are located. We confirm this argument more concretely with an explicit expression for such operator spread coefficients in \appref{App:cmunu}.

The above statement has important consequences for the OTOC. If site $s$ lives in the support of $\vec{\mu}$ then the contribution to Eq.~\eqref{eq:OTOC_from_c} coming from the strings with the same support as $\vec{\mu}$ have an equal weight for each possible value $\theta_{\mu,Z_{s}} \in \frac{2\pi}{q}\{1,\ldots,q\}$, so that the cosine term averages to zero. The remaining part is the total weight due to Pauli strings which are supported on intervals that do not contain site $s$ (along with some corrections for Pauli strings which border on site $s$). Deferring the full justification \appref{app:otoc}, the upshot is that provided $\kappa>0$ in the $\tau\rightarrow\infty$ limit, the OTOC behaves as
\begin{equation}\label{eq:OTOCtoR}
\overline{\mathcal{C}}(s,\tau) \approx 1-\overline{R}\left(s-1,\tau\right) +  q^{-2}\overline{\rho_{R}}\left(s,\tau\right)\punc{.}
\end{equation}
up to exponentially small corrections in $\tau$. Hence the OTOC physics is directly related to the operator density, and changes appreciably at the operator front $s=v_{\text{B}}t$, as we show in \figref{fig:random_circuit_otoc}.

Let us now summarize the behaviour of the OTOC as a function of space and time, as parameterized by the ratio $\kappa \equiv s/\tau$ and taken in the limit $\tau\rightarrow\infty$. We  distinguish four regimes of OTOC behaviour which we illustrate in \figref{fig:random_circuit_otoc}. 

\begin{enumerate}

\item OTOC trivial at early times ($1< \kappa$): In this regime the events $(\tau,0), (0,s)$ are causally disconnected, so the commutator in \eqnref{eq:OTOC_def} (and hence the OTOC) is exactly zero. 

\item Early OTOC growth ($v_{\text{B}} < \kappa <1 $): This regime describes the behaviour after site $s$ has entered the light cone, but before it encounters the main operator front. Here we approximate the OTOC using \eqnref{eq:Rapprox}, so that $\overline{\mathcal{C}}(s,\tau) \approx c_1 \overline{\rho_{R}}(s-1,\tau)+\overline{\rho_{R}}(s,\tau)$, where $c_1>1$ is bounded in the $s,\tau\rightarrow \infty$ limit. Fortunately a simple closed form expression already exists for $\overline{\rho_{R}}$, namely \eqnref{eq:rho_realspace}. We obtain a more convenient expression for the initial OTOC growth by expanding \eqnref{eq:rho_realspace} near the light cone in the $\delta^2 /s \rightarrow 0$ limit, where $\delta\equiv \tau-s$  is the distance between $s$ and the light cone
\begin{equation}\label{eq:lyaponov}
\overline{\mathcal{C}}(s,\tau) \approx   e^{\frac{1}{2} \delta  \log (\frac{\gamma s }{\delta })-\frac{1}{6 \delta }} \times \left(\frac{q^2}{q^2+1}\right)^s  \frac{(1+q^2)\sqrt{\delta}} {2s\sqrt{\pi   }} .
\end{equation}
up to multiplicative O(1) constants, where $\gamma=e\left(1-v_{\text{B}}^{2}\right)/2$. This formula demonstrates that the OTOC will increase with an exponent $\lambda \sim \log  s$ for $0<\delta \ll s$. Due to its dependence on $s$, and its limited range of validity, it is unclear whether this should be viewed as a Lyapunov exponent as in Ref. \onlinecite{Maldacena2016}. Note that the exponential increase occurs in a regime where the overall scale of the OTOC is still exponentially small in $s$. In the regime where the OTOC increases to an $O(1)$ value (i.e., when the operator front hits, see next point) its behavior is \emph{not} exponential. Furthermore, we note that $\gamma\sim 1/q^{2}$ for large $q$, such that the regime in which the exponential behavior can be observed becomes \emph{smaller} in the large $q$ limit.

\item Near the front ($ |\kappa -  v_{\text{B}}| =  O(1/\sqrt{\tau})$): As mentioned, the above approximation breaks down when the main front, which we recall travels at speed $v_{\text{B}}$ and has width $\sim \sqrt{\tau}$, arrives at site $s$. In this intermediate regime, we estimate the OTOC by combining \eqnref{eq:OTOCtoR} and \eqnref{eq:frontRscaling}
\begin{equation}\label{eq:otoc_intermediate}
\overline{\mathcal{C}}(s,\tau) \approx\frac{1}{2} \text{erfc}\left(\frac{s-v_{\text{B}}\tau}{\sqrt{2\tau(1-v_{\text{B}}^{2})}}\right)\punc{.}
\end{equation}
This formula describes the behaviour of the OTOC in the regime when it increases from a value exponentially small in $s$ to an $O(1)$ number.

\item Late times ($0<\kappa< v_{\text{B}}$): After the main front has passed the OTOC relaxes exponentially to $1$. Expanding Eq.~\eqref{eq:Rapprox} for fixed $s-v_{\text{LC}} \tau$ and large $\tau$ we find that the OTOC in this late time regime is 
\begin{equation}\label{eq:otoc_late}
1-\overline{\mathcal{C}}(s,\tau)  \approx\frac{(1+q^{-2})q^{s}}{\sqrt{8\pi \tau}}\left(\frac{2q}{1+q^{2}}\right)^{\tau}
\end{equation}
Thus the OTOC decays to its equilibrium value with an exponent $\log(\frac{1+q^{2}}{2q})$.

\end{enumerate}
 \begin{figure}[h!]
 \centering
  	\includegraphics[width=0.38\textwidth]{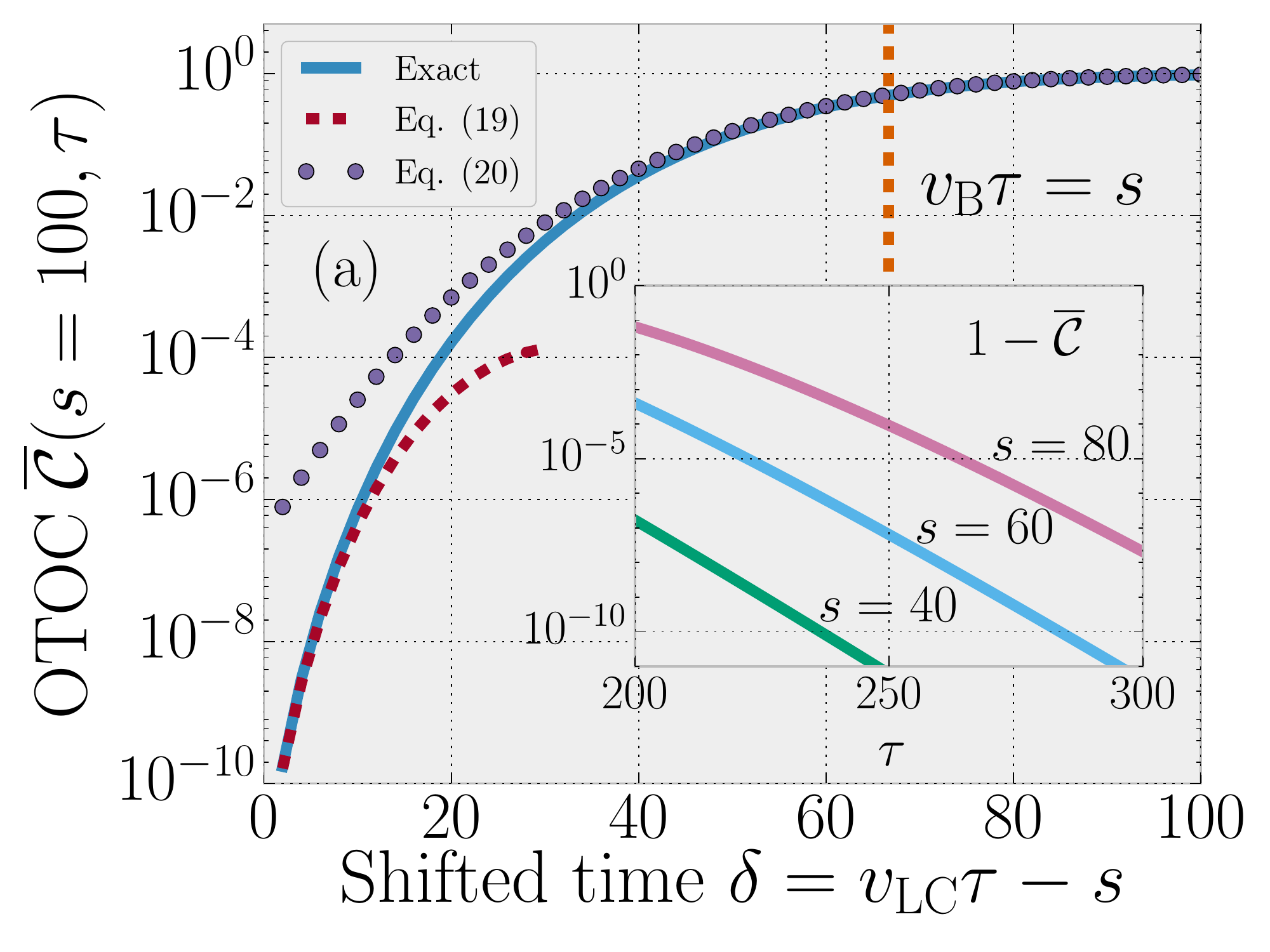} 
  	\includegraphics[width=0.38\textwidth]{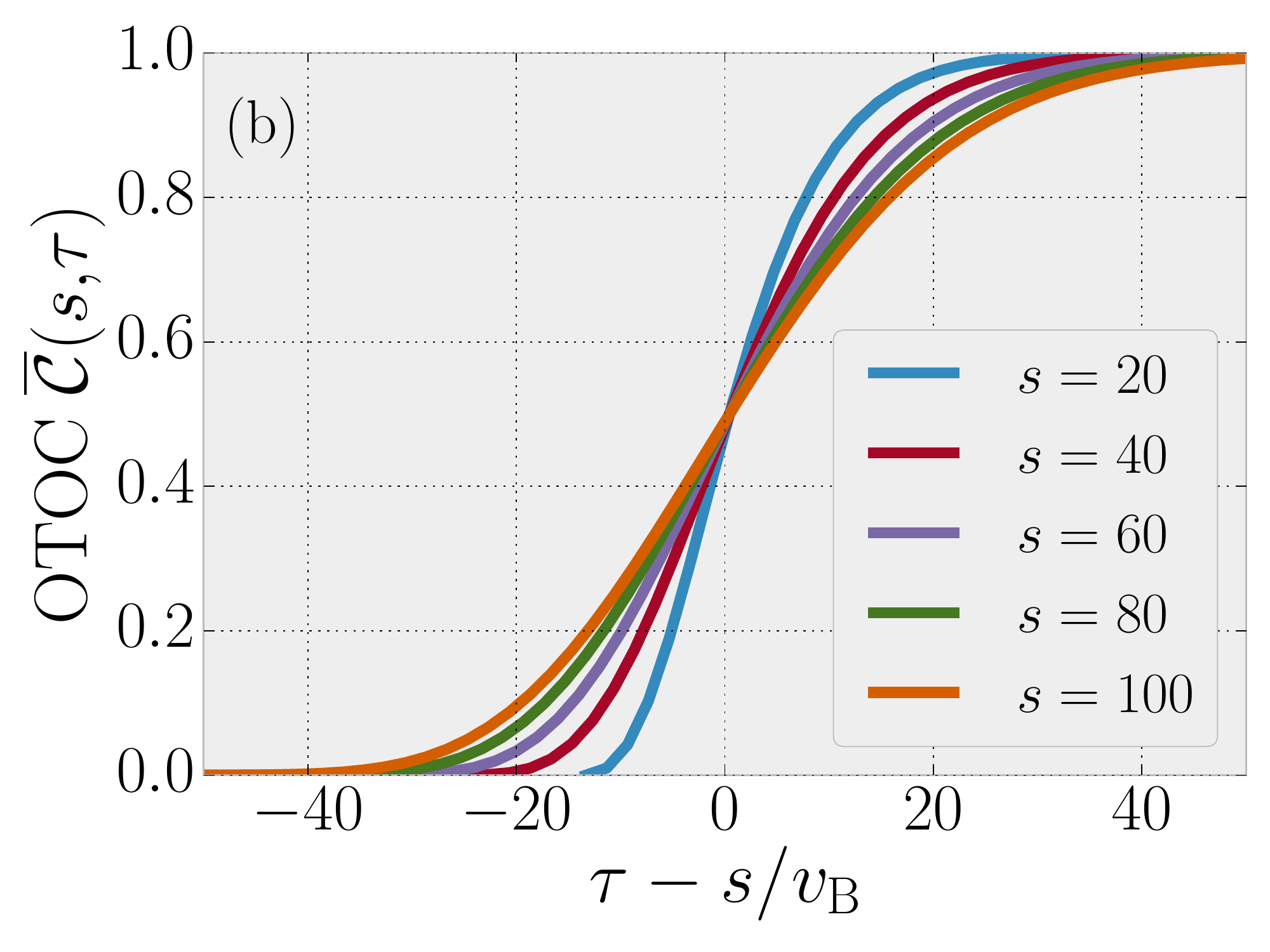}
\caption{Time dependence of the average OTOC in the random circuit model. (a) Different time regimes for fixed separation $s=100$. The exact result for the OTOC follows Eq.~\eqref{eq:lyaponov} after the light cone hits site $s$. The behaviour than goes over to regime described by Eq~\eqref{eq:otoc_intermediate} after the front with speed $v_{\text{B}}$ arrives. The inset shows the exponential decay of the OTOC to its final value $1$, as described by Eq.~\eqref{eq:otoc_late}, for different separations. (b): scaling collapse of the OTOC at the front.}
 \label{fig:random_circuit_otoc}
 \end{figure}
 
\subsection{Fluctuations from circuit to circuit}\label{ss:fluct}

The results discussed above concern quantities averaged over many different random circuits with the same geometry but different choices of two-site gates. The question remains regarding whether these average quantities are also `typical', i.e. how large are the fluctuations between different realizations of the random circuit. In this section we investigate this problem numerically. Our numerical method relies on representing the operator $Z_0(t)$ as a matrix product operator (MPO), which allows us to apply the two-site unitary gates efficiently. Two layers of the random circuit can be applied by just a single step of the TEBD algorithm, which allows us to go up to bond dimension $\chi = 20000$. Both the infinite temperature OTOC and the total operator weight contained in an arbitrary subregion can be extracted straightforwardly from the MPO representation (both calculations are similar to computing the overlap of two matrix product states, but in the computation of $R(s)$ only the legs corresponding to sites $\leq s$ are contracted).

To quantify the fluctuations we look at an ensemble of 100 random circuit realizations (which is enough to reliable reproduce the exact average quantities, see~\figref{fig:fluctuations}) with on-site Hilbert space dimension $q=2$ and compute a) the OTOC $\mathcal{C}(s,\tau) = 1-\mathrm{Tr}[(Z_0(\tau)Z_s)^2 ]/ 2^L$ and b) the total operator weight $R(s,\tau)$ of $Z_0(\tau)$ contained within the region left of site $s$. Both $R(s,\tau)$ and $\mathcal{C}(s,\tau)$ are functions of the distance $s$ and the number of layers $\tau$. We find that for both quantities, the circuit-to-circuit fluctuations are largest at the traveling wavefront and become smaller deep behind it. This is shown in Fig.~\ref{fig:fluctuations}. This also shows that there is a well defined front for the information propagation in each individual circuit.
 
We also find that the the fluctuations decrease in time. Fig.~\ref{fig:fluctuations} (c) shows the standard deviation of the weight $R(s)$ for different times. We find that the maximum of this standard deviation over all values of $s$ decreases in time, approximately as $\propto \tau^{-\beta}$ with an exponent $0.4<\beta<0.5$. 

 \begin{figure}[h!]
 \centering
  	\includegraphics[width=0.23\textwidth]{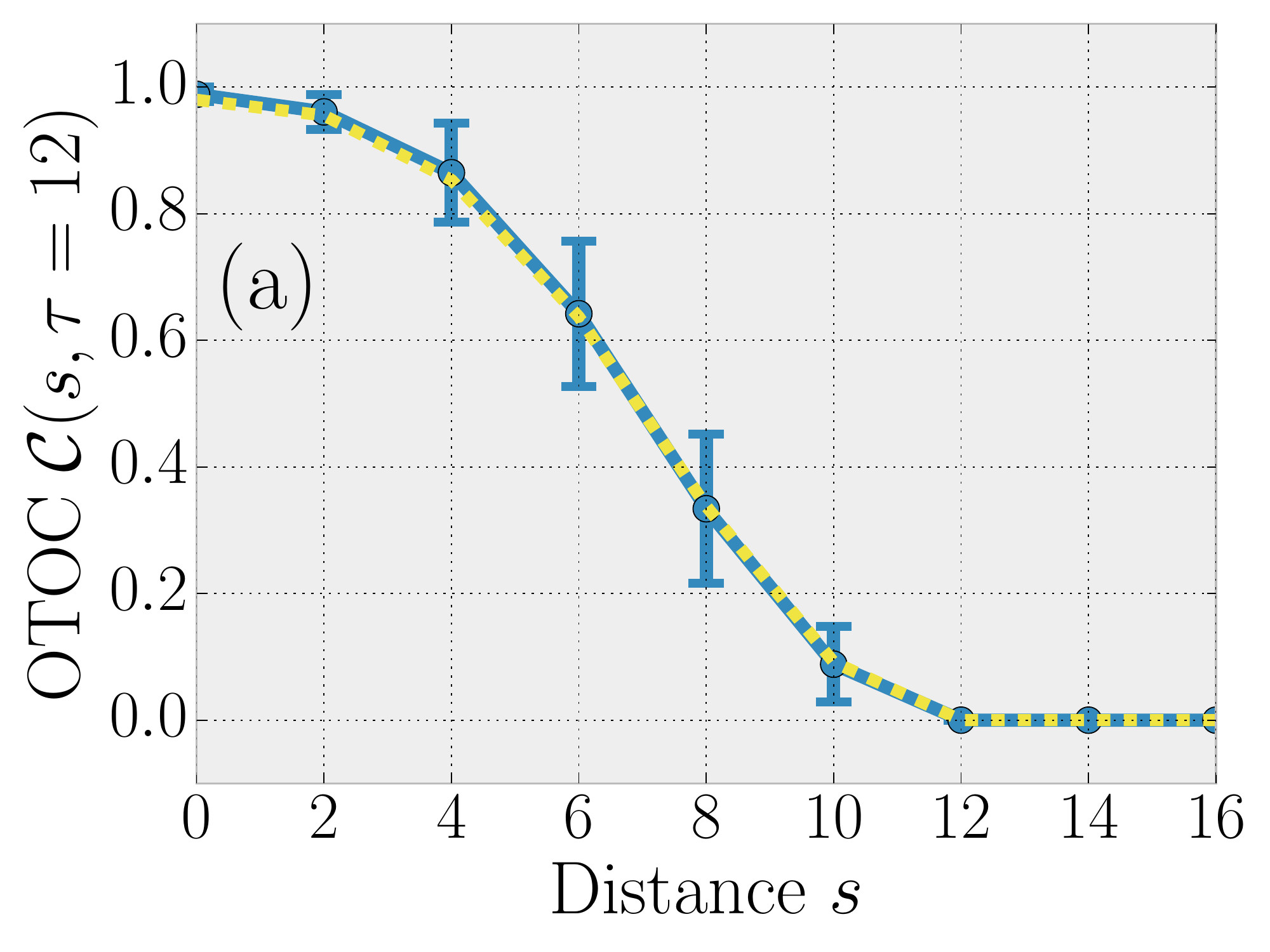}
  	\includegraphics[width=0.23\textwidth]{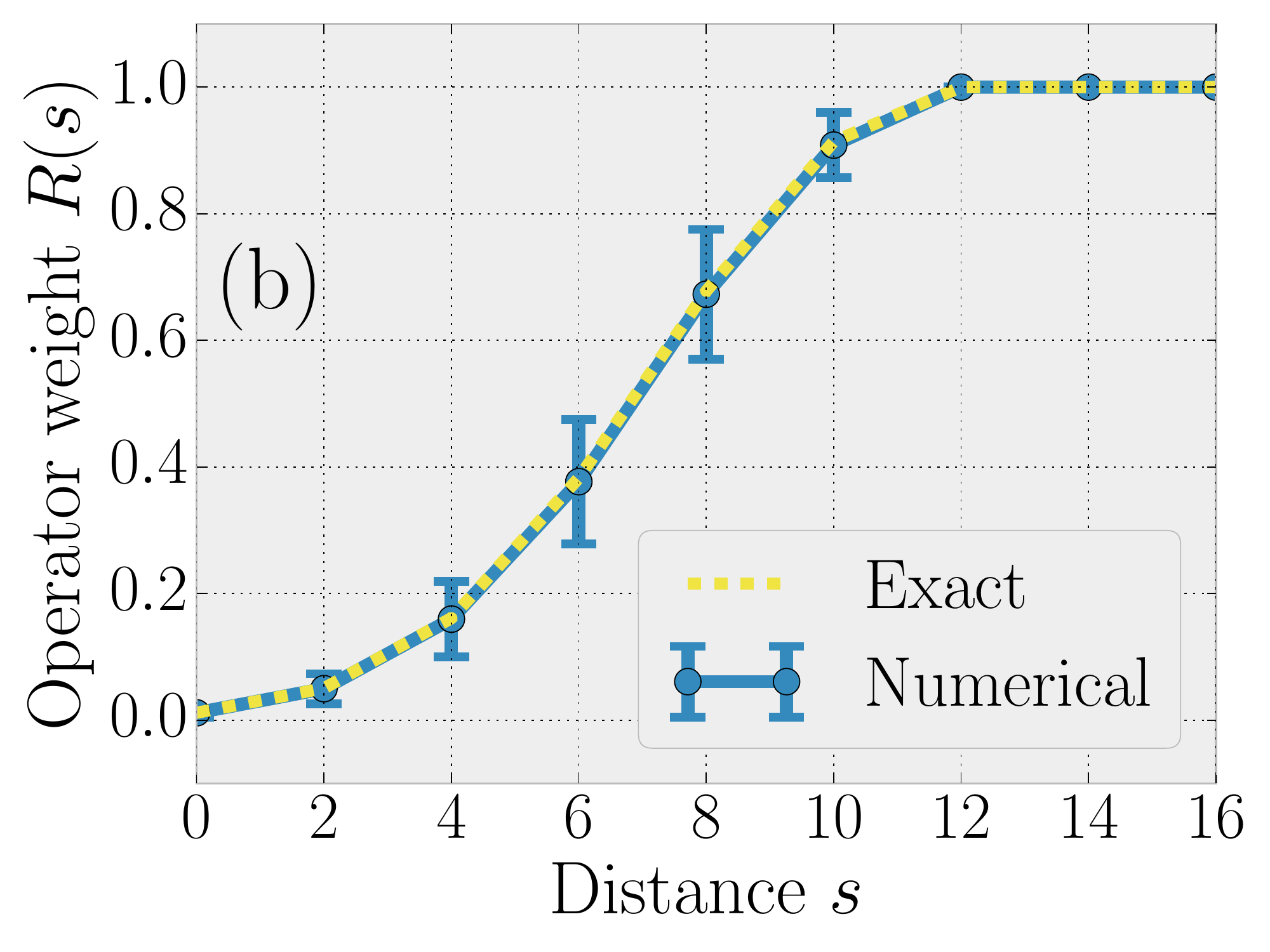}
  	\includegraphics[width=0.35\textwidth]{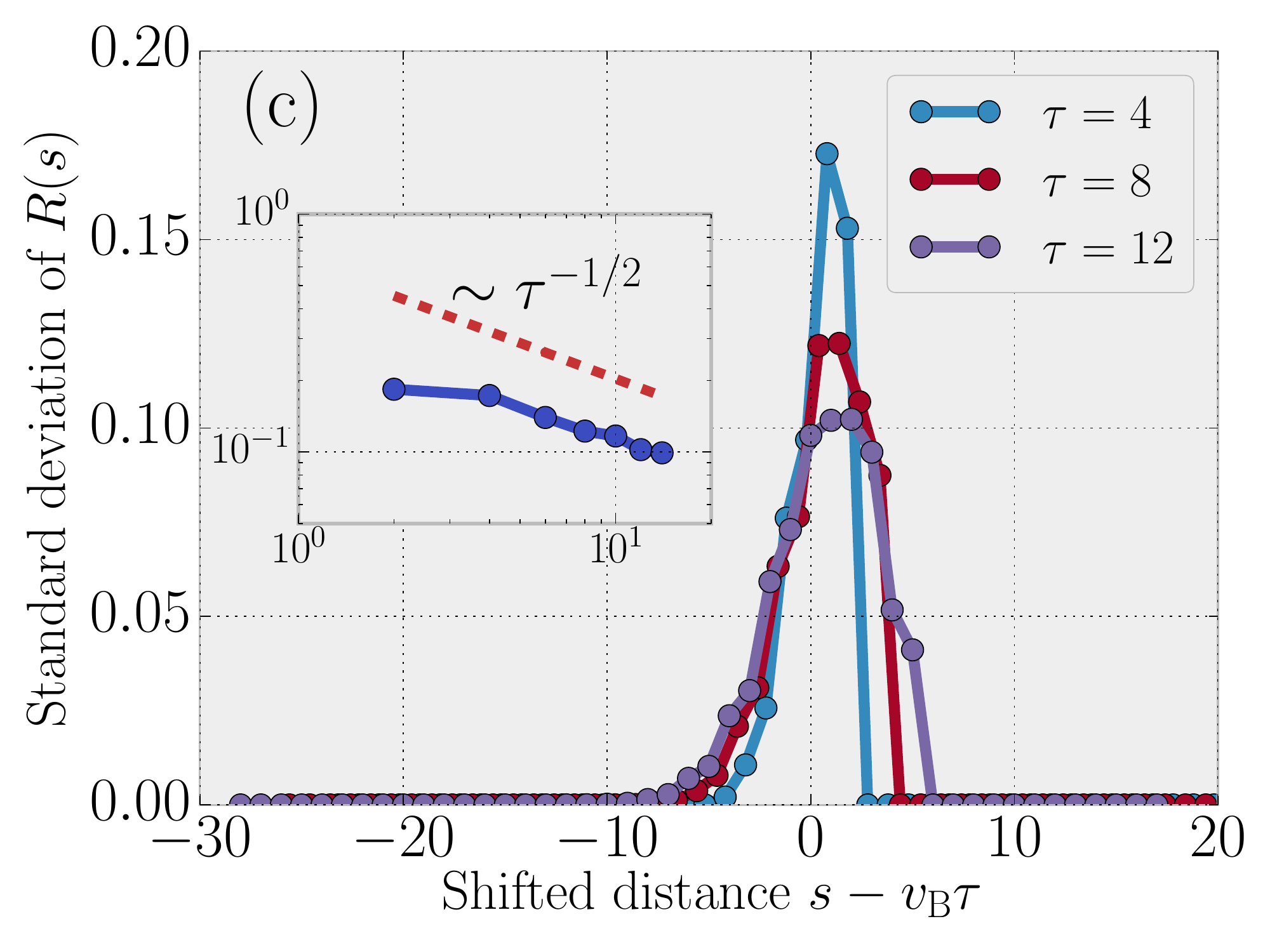} 
\caption{Average values and fluctuations of the (a) OTOC and (b) the total weight left of site $s$ for the time evolved operator $Z_0(\tau)$ after $\tau=12$ layers of the random circuit. Blue dots correspond to average values of 100 different random circuits while the error bars signify one standard deviation. Figure (c) shows the standard deviations of $R(s)$ for different times. The largest fluctuations decrease in time approximately as $\propto \tau^{-1/2}$, as shown by the inset.}
 \label{fig:fluctuations}
 \end{figure}

\subsection{Relationship to entanglement spreading}\label{ss:randomcircuit_entanglement}

Another question closely related to operator hydrodynamics is the problem of entanglement growth. An initial product state develops spatial entanglement during time evolution. In systems without quenched disorder the entanglement is expected to grow linearly in time, with a growth rate characterized by the `entanglement velocity' $v_{\text{E}}$. In this section we use our results for operator spreading to compute the typical value of the second R\'enyi entropy between two sides of a spatial entanglement cut and extract the entanglement velocity from it. We find that this velocity is smaller than the butterfly velocity for any finite $q$ and approaches the light cone velocity logarithmically slowly, so that $v_\text{LC} - v_{\text{E}} \propto 1 / \log{q}$ for large $q$. At long times the R\'enyi entropy saturates to its maximal value with the saturation becoming increasingly sharp as $q$ is increased.

Consider an initial ferromagnetic product state of the 1D chain where the state on site $s$ is an eigenstate of the local Pauli operator $Z_s$ with eigenvalue $+1$ (Note that for the average behavior of the random circuit the choice of initial product state is unimportant). The density matrix $\hat\omega$ corresponding to this state is then a sum over all possible \emph{$Z$-strings}, i.e. Pauli strings that only contain powers of the operator $Z$ on each site:
\begin{equation}\label{eq:Z_prod_state}
\hat\omega = \frac{1}{q^L}\prod_{s=1}^{L} \left(\sum_{k=0}^{q-1}Z_s^k\right) = \frac{1}{q^L} \sum_{\vec{\nu}\in\text{Z-strings}}\sigma^{\vec{\nu}}.
\end{equation}
The density matrix at time $t$ is obtained by replacing each Pauli string $\sigma^\nu$ in Eq.~\eqref{eq:Z_prod_state} with its time evolved counterpart $\sigma^{\vec{\nu}}(t)$.

Let us now divide the system into two regions, $A$ and $B$, the first of which corresponds to sites $1,\ldots,L_\text{A}$. Generalizing the formula of Refs.~\onlinecite{Abanin17,Mezei16}, the second R\'enyi entropy $S^{(2)} = -\log\text{tr}(\hat\omega_\text{A}^2)$ of the reduced density matrix $\hat\omega_\text{A} = \text{tr}_\text{B}(\hat\omega)$ is related to the operator spreading coefficients by
\begin{align}\label{eq:S_Renyi_from_c}
e^{-S^{(2)}} = \frac{1}{q^{L_\text{A}}} \sum_{\vec{\nu},\vec{\nu'}}\sum_{\vec{\mu}\subset A} c_{\vec{\mu}}^{\vec{\nu}} {c_{\vec{\mu}}^{\vec{\nu'}}}^*
\approx \frac{1}{q^{L_\text{A}}} \sum_{\vec{\nu}}\sum_{\vec{\mu}\subset A} |c_{\vec{\mu}}^{\vec{\nu}}|^2
\end{align}
where the strings $\vec{\nu}$ and $\vec{\nu'}$ are both $Z$-strings and $\mu$ has support entirely in subsystem $A$. In the last equality of Eq.~\eqref{eq:S_Renyi_from_c} we assumed that the off-diagonal contributions are negligible, which becomes \textit{exactly} true in the random circuit model once we average over different realizations.

Let us assume that $L_\text{A}$ is even. Reverting back to the coarse-grained position $x$ (see Eq.~\eqref{eq:newrho}) we recognize Eq.~\eqref{eq:S_Renyi_from_c} as the total operator weight in region $A$, $R(x=L_{\text{A}}/2,t)$, as defined in Eq.~\eqref{eq:Rdefn}, summed over all initial $Z$-strings. As we noted previously, this quantity is on average the same for all initial strings with the same endpoints $x_0$. The number of different $Z$-strings with right endpoint $x_0$ is $q^{2(x_0-1)} (q^2-1)$. After averaging over random circuits, and assuming an even number of layers, Eq.~\eqref{eq:S_Renyi_from_c} thus becomes
\begin{equation}~\label{eq:S_Renyi_from_Rho}
\overline{e^{-S^{(2)}(\tau)}} = \frac{1}{q^{L_\text{A}}} + \frac{q^2-1}{q^2}\sum_{x_0 = 1}^{L/2} \frac{\overline{R}(L_\text{A}/2-x_0,t=\tau/2)}{q^{L_\text{A}-2x_0}},
\end{equation} 
where we have used that $\overline{\rho_R}(x)$ (and consequently $\overline{R}(x)$) only depends on the position $x$ relative to the initial endpoint $x_0$. The first term in Eq.~\eqref{eq:S_Renyi_from_Rho} is the contribution coming from the identity operator, which is responsible for the saturation of the entanglement at long times.

Using the exact solution Eq.~\eqref{eq:rho_solution} one can perform the sum over initial positions to find
\begin{equation}\label{eq:S_Renyi_exact}
\overline{e^{-S^{(2)}(\tau)}} = q^{-L_\text{A}} + [1 - q^{-L_\text{A}}] \frac{q^\tau}{\left(1+q^2\right)^\tau} \sum_{x=-\tau/2}^{L_\text{A}/2 -1} {\tau \choose \frac{\tau}{2} + x}.
\end{equation}
The sum over binomial coefficients can be expressed in terms of a hypergeometric function.

Eq.~\eqref{eq:S_Renyi_exact} describes an entanglement that initially increases linearly with time and saturates to the maximum value $L_\text{A}\log{q}$ at long times. For $\tau \ll L_\text{A}$ we find
\begin{equation}
\overline{e^{-S^{(2)}(\tau)}} \approx \left( \frac{2q}{1+q^2} \right)^\tau,
\end{equation} 
from which we can identify the \emph{entanglement velocity}
\begin{align}\label{eq:entanglement_velocity}
v_\text{E} \equiv \frac{1}{\log{q}} \frac{\text{d}S(\tau)}{\text{d}\tau} &=\frac{\log \frac{q+q^{-1}}{2}}{\log q} = \frac{\log \left(1-v_{\text{B}}^2\right)}{\log \left(\frac{1-v_{\text{B}}}{1+v_{\text{B}}}\right)}.
\end{align}
Note that the entanglement velocity approaches $1$ logarithmically slowly for large $q$, i.e., $v_E\sim 1- \log(2)/\log(q)$. This is a separate velocity scale, distinct from, and smaller than the front speed $v_\text{E} <v_{\text{B}}$. This difference comes from the diffusive broadening of the operator wavefront.  First, it is straightforward to verify that if the wavefront is sharp, i.e., $\overline{R}(x,t) = \Theta(x-v_{\text{B}}t)$ then Eq.~\eqref{eq:S_Renyi_from_Rho}  gives $v_\text{E} = v_{\text{B}}$. Second, we have checked that \eqnref{eq:S_Renyi_from_Rho} gives $v_\text{E} = v_{\text{B}}$ even if the wavefront has a width which is finite but independent of time\footnote{We checked this for front profiles arising from $\rho_R(x,t) \propto e^{-(x-v t)^2/\sigma^2}$ and $\rho_R(x,t) \propto e^{-|x-v t|/\sigma}$.}. Hence, we attribute the difference between $v_{\text{B}}$ and $v_\text{E}$ to the fact that the operator front broadens in time. 

In the right panel of \figref{fig:randomcircuit_entanglement} we compare the exact formula Eq.~\eqref{eq:S_Renyi_exact} to the second R\'enyi entropy as computed numerically (using a matrix product state representation), averaging over $100$ realizations of the circuit and find extremely good agreement. Moreover, the numerical calculation allows us to compare the typical and average values of the R\'enyi entropy, defined as $S^{(2)}_\text{typ} = -\log{\overline{e^{-S^{(2)}}}}$ and $S^{(2)}_\text{avg} = \overline{S^{(2)}}$, respectively. We find no significant difference between the two values, showing that there are no strong circuit-to-circuit fluctuations in the entropy and both are captured well by our exact formula. We also found numerically that replacing the R\'enyi entropy with the von Neumann entropy leads to a slightly larger entanglement velocity.

The entanglement saturates when the contribution of the identity becomes significant (i.e. when all other operators have essentially left the subsystem). Note that the saturation softens, compared to the prediction of the simple operator spreading model of Ref. \onlinecite{Abanin17}, which is another consequence of the diffusive broadening of the front. This intermediate saturation regime becomes smaller with increasing $q$, as shown in the left panel of~\figref{fig:randomcircuit_entanglement}.

 \begin{figure}[h!]
 \centering
  	\includegraphics[width=0.23\textwidth]{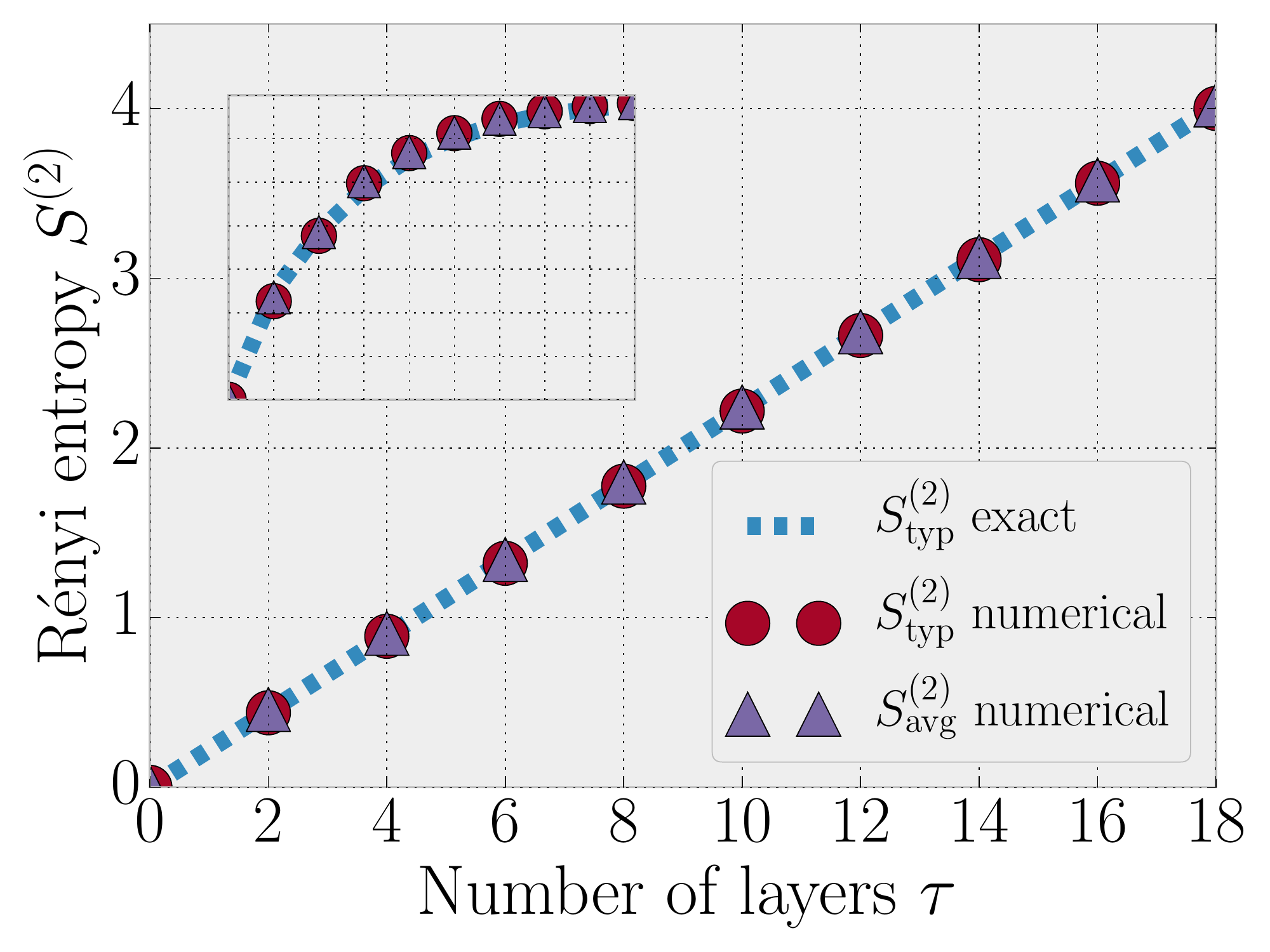}
  	\includegraphics[width=0.23\textwidth]{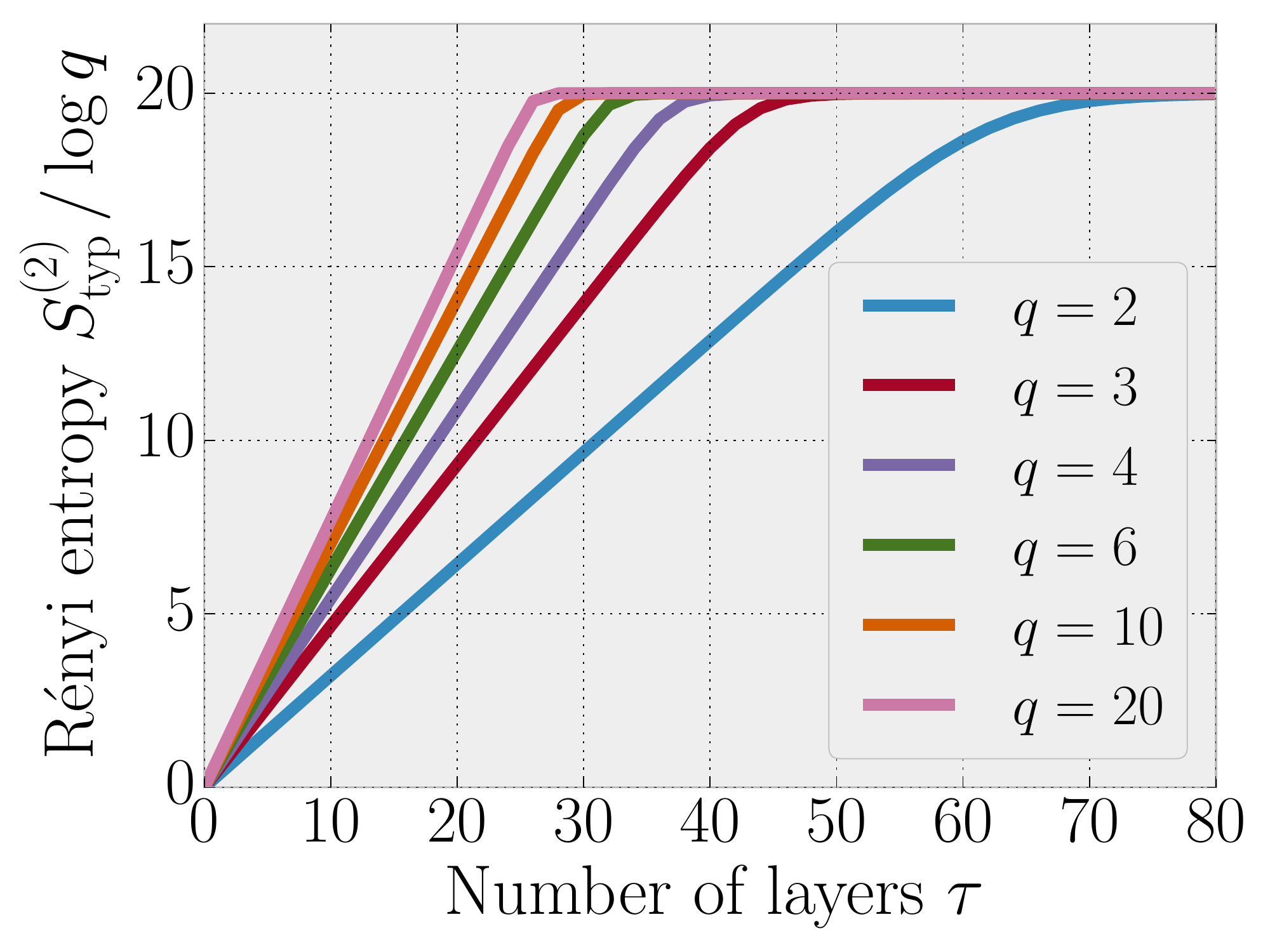} 

 \caption{Entanglement growth in the random circuit model. Left: comparing the exact formula, Eq.~\eqref{eq:S_Renyi_exact}, to matrix product state numerics shows that it captures both the typical value $S^{(2)}_\text{typ} = -\log{\overline{e^{-S^{(2)}}}}$, and the average $S^{(2)}_\text{avg} = \overline{S^{(2)}}$ of the second R\'enyi entropy. The main figure shows the time dependence for $L_\text{A} = 50$ sites, while the inset is for $L_\text{A} = 2$. Right: The entanglement velocity increases with $q$ according to Eq.~\eqref{eq:entanglement_velocity} while the saturation regime becomes smaller.}
  \label{fig:randomcircuit_entanglement}
 \end{figure}

\section{Comparison with the kicked Ising model}\label{s:kickedIsing}

A natural question that emerges in relation to the results stated above, is to what extent are they representative of other, more generic thermalizing quantum many-body systems. To address this question we investigate a system with a periodically driven nearest neighbor Hamiltonian. Our model has the same geometry as the random circuit, shown in~\figref{fig:circuit} and it similarly does not conserve energy. However, unlike random circuits, it is periodic in time and its two-site (and one-site) gates take a specific form, rather than randomly chosen. Despite these two significant differences, we find that several details of the operator spreading described in \secref{s:randomcircuit}, such as the diffusive broadening of the wavefront, remain approximately valid.

For concreteness, we consider a model with on-site Hilbert space dimension $q=2$ that consists of switching back and forth between two Hamiltonians, such that one period of the time evolution (with period time $T$) is given by
\begin{equation}\label{eq:Ising_def}
\hat U = e^{-i \frac{T}{2} h \sum_s X_s} e^{-i\frac{T}{2}\sum_s[Z_sZ_{s+1} + gZ_s]}.
\end{equation}
This system \textemdash which we refer to as the `kicked Ising model' \textemdash is known to be ergodic, provided that both $g$ and $h$ are sufficiently large. Since at any given time the terms in the Hamiltonian all commute with each other, the time evolution can be represented as a circuit of two-site unitaries (with the one-site rotations included in the two-site gates) with the same geometry as in Fig.~\ref{fig:circuit}. As such, it is in fact contained among the ensemble of random circuits considered before. The question is to what extent do the properties of this specific circuit coincide with the average quantities discussed above.

At first, operator spreading in the Floquet system seems very different from the case of the random circuit. There is no inherent randomness and the evolution is completely determined by the internal structure of the initial operator $\sigma^\mu$, while for the random circuit the average behavior was independent of the internal structure. The random unitary and Floquet systems do have something in common, however: both involve evolution under a \textit{local unitary} circuit. This puts strong constraints on the evolution of $\rho_R$: not only does it obey global conservation law $\int \text{d}x \rho_{R}(x,t) = const.$ (we revert to a continuum notation for ease of presentation), it should also obey a local conservation law 
\begin{equation}
\partial_t \rho_{R}(x,t) +\partial_x J(x,t) = 0
\end{equation}
for some local current density $J(x,t)$. This conservation law puts severe restrictions on the equation of motion of $\rho_R$. For example, one can imagine that in a coarse grained picture, on long enough time scales, the constitutive equation $J\sim v \rho_R + D \partial_x \rho_R +\ldots$ becomes valid; note that the discretized version of this constitutive relation is exactly the random walk equation we derived for the random circuit averaged $\rho_R$, see \eqnref{eq:random_walk}.  Therefore it seems plausible that in a sufficiently coarse grained picture, the dynamics might be well approximated by a biased diffusion similar to the one described in \secref{s:randomcircuit} for the Floquet circuit, with hopping probabilities depending on the microscopic couplings. Here we present numerical evidence in support of this conjecture. Our results can be summarized in three points:
\begin{itemize}
\item The butterfly velocity $v_{\text{B}}$ depends strongly on the coupling $g$ and can be tuned to be much smaller than the light cone velocity $v_\text{LC}$
\item When tuning the couplings to decrease $v_{\text{B}}$ from its maximal value $v_{\text{B}}\approx v_\text{LC}$, the front also becomes wider, as expected for a random walk when increasing the probability of hopping to the left at the expense of the probability of hopping right.
\item The operator wavefront gets wider during time evolution, with the width increasing in time as $\sim t^\alpha$, with an exponent $0.5\lesssim\alpha\lesssim 0.6$
\end{itemize}

We find numerically a linearly propagating wavefront for the time evolved operator $Z_0(t)$, which shows up in both the OTOC and the weight $R(s)$, with the OTOC $\mathcal{C}(s,t)$ saturating to 1 behind the front. While for the random circuit the speed of the front was set by the on-site Hilbert space dimension $q$, for the kicked Ising model we find that this speed can be tuned continuously by changing the value of the transverse field $g$~\footnote{we found that the butterfly velocity does not depend significantly on the value of $h$, provided that it is not too small}, as shown in Fig.~\ref{fig:Ising_tunable_vB}. Note that changing $g$ does not affect the light cone velocity, which is $\Delta s / \Delta T = 1$ due to the geometry of the circuit that represents the Floquet time evolution. For $g=0$, an initial operator $Z_0$ remains localized on the same site. As we make $g$ larger, the butterfly velocity gradually increases and it reaches $v_{\text{B}}\approx v_\text{LC}$ for $g\approx 0.9$ with period time $T=1.6$~\footnote{we note that this set of parameters, for which $v_{\text{B}}$ is near maximal, is close to the parameter choice found in Ref.~\onlinecite{Huse16} to exhibit very fast thermalization}. Looking at~\figref{fig:Ising_tunable_vB} we notice that decreasing $v_{\text{B}}$ from its maximum corresponds to an increased front width at any given time. This is consistent with a coarse-grained random walk description, wherein increasing the probability of hopping to the right results in both a larger butterfly velocity and a suppression of the diffusion constant.

 \begin{figure}[h!]
 \centering
  	\includegraphics[width=0.23\textwidth]{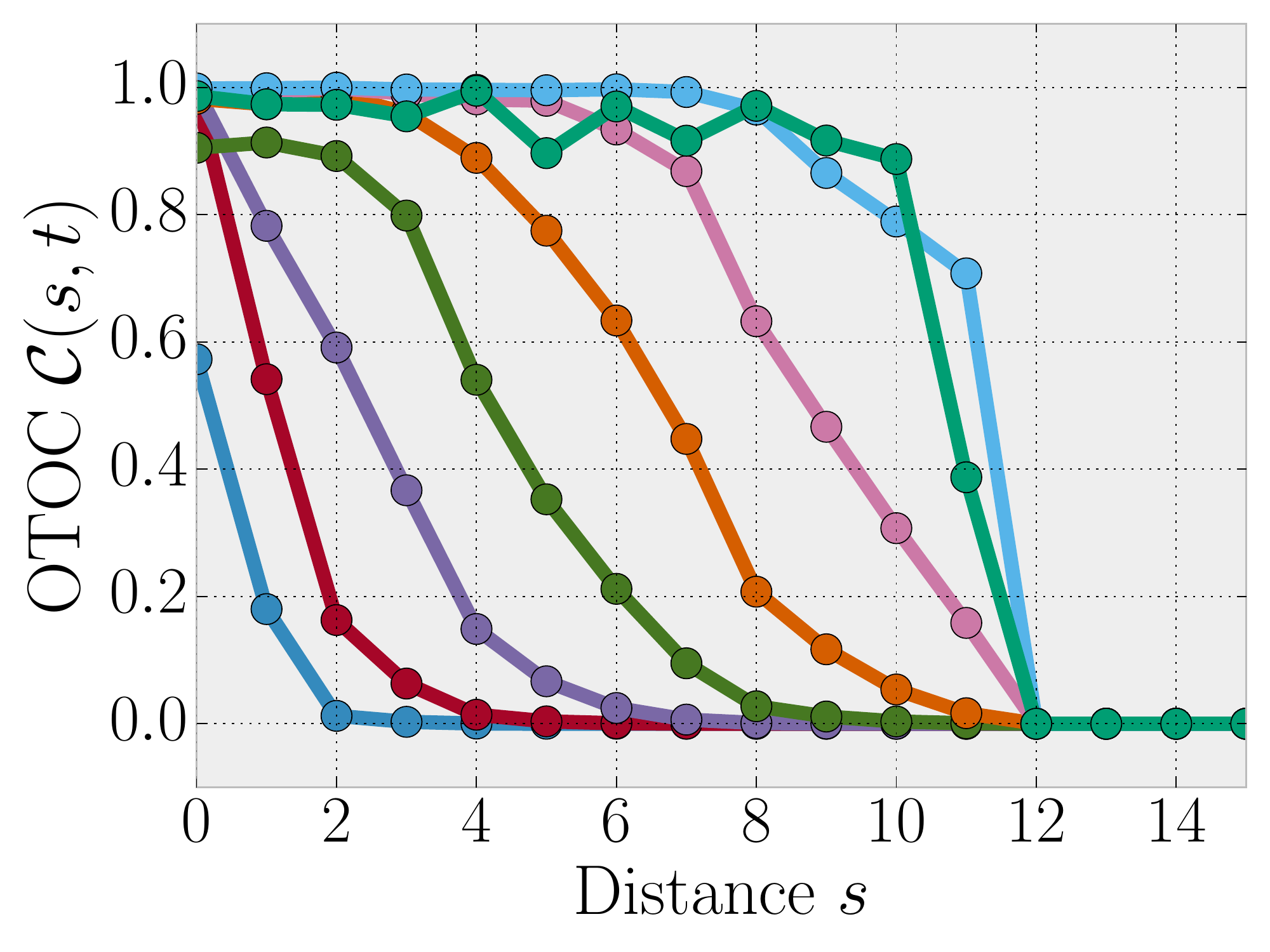}
  	\includegraphics[width=0.23\textwidth]{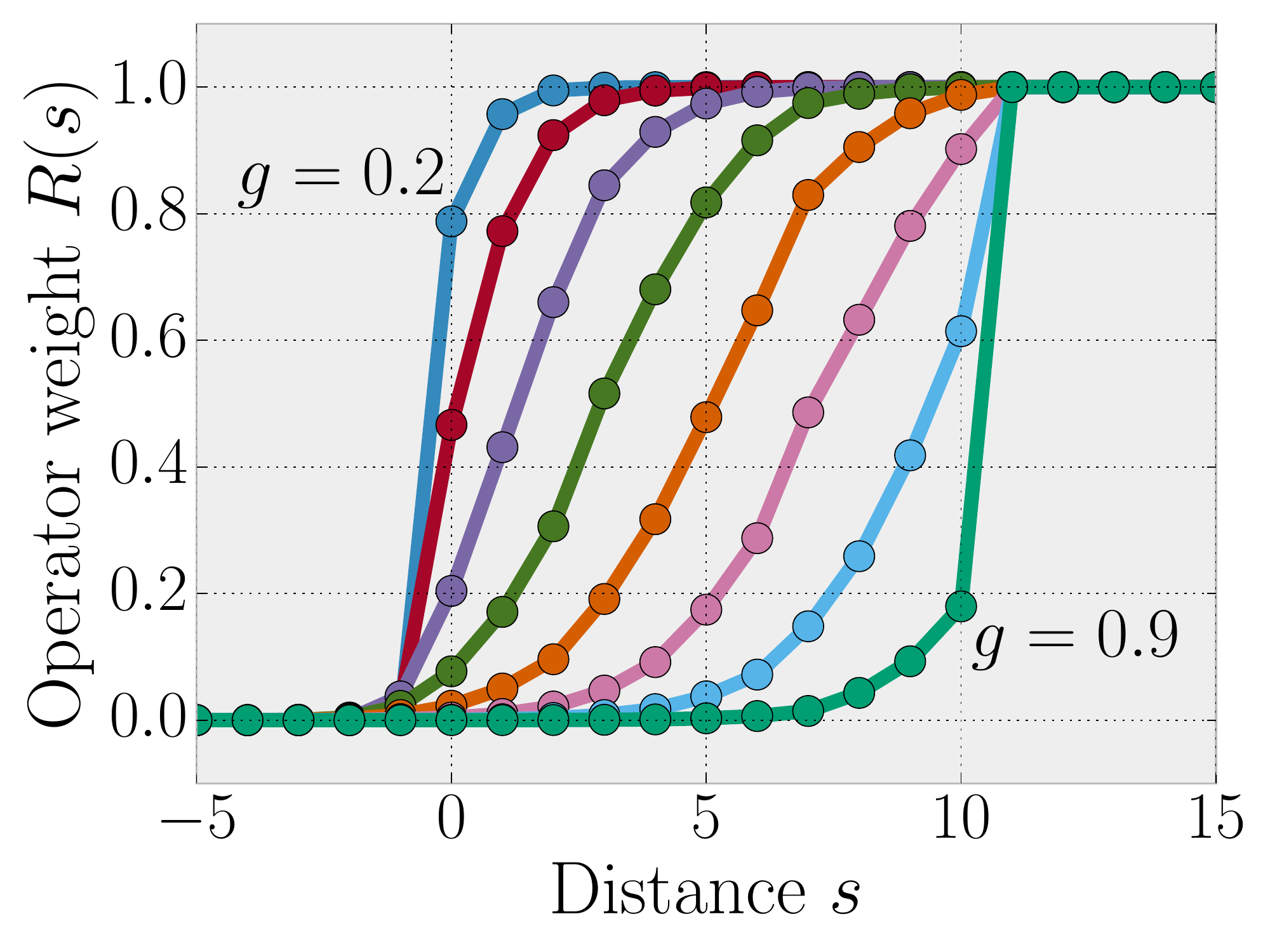} 
\caption{OTOC (left) and operator weight $R(s)$ (right) for different distances $s$ after $t/T = 12$ driving cycles of the kicked Ising model with the strength of the transverse field $g=0.2, 0.3, \ldots, 0.9$ as indicated in the right figure. The longitudinal field is fixed at $h=0.809$ while the period time is $T = 1.6$. The butterfly velocity shows a strong dependence on the coupling $g$, with the front width increasing as one moves away from the limit of maximal velocity.}
 \label{fig:Ising_tunable_vB}
 \end{figure}

The most important evidence in support of a hydrodynamic description comes from examining the front width as a function of time. Similarly to the random circuit model, we find the wavefront of the operator spreading broadens as we go to longer times. To quantify the width by looking at the standard deviation of $\rho_R(s) = R(s) - R(s-1)$, i.e.
\begin{equation}\label{eq:frontwidth}
\sigma(t) \equiv \sqrt{\sum_s \rho_R(s) s^2 - \left[\sum_s \rho_R(s) s\right]^2}
\end{equation}
As shown in Fig.~\ref{fig:KimHuse_front_widening}, at long times the width grows algebraically in time as $\sigma(t) \propto t^\alpha$ with an exponent $0.5\lesssim\alpha\lesssim 0.6$. This is roughly consistent with the random walk description of operator spreading put forward in \secref{s:randomcircuit}. This diffusive broadening is expected to result in the strict inequality $v_\text{E} < v_\text{B}$ for the entanglement velocity, according to the arguments put forward in Sec.~\ref{ss:randomcircuit_entanglement}. We confirmed numerically that this indeed holds in this model for various values of $g$.

 \begin{figure}[h!]
 \centering
  	\includegraphics[width=0.23\textwidth]{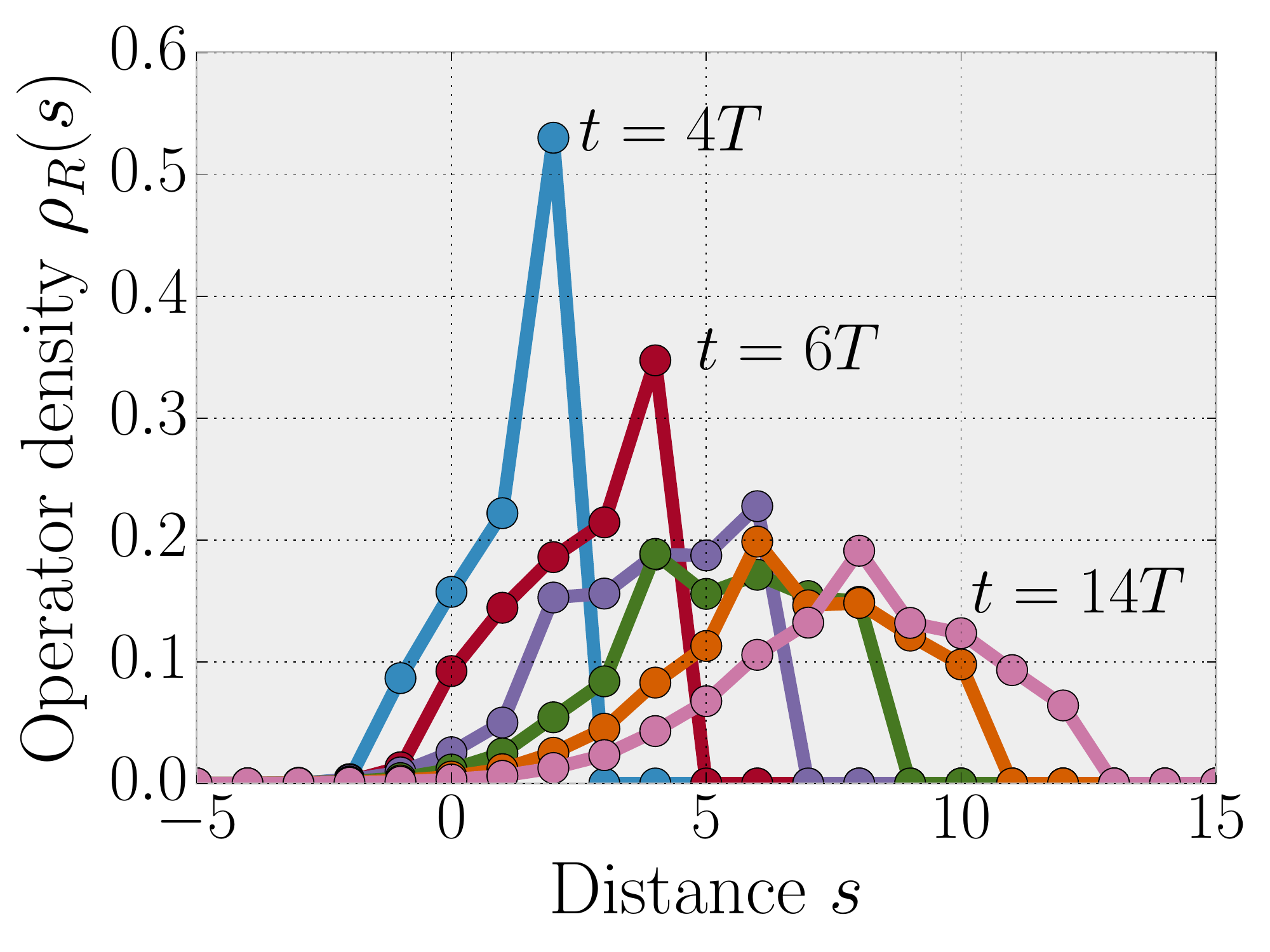}
  	\includegraphics[width=0.23\textwidth]{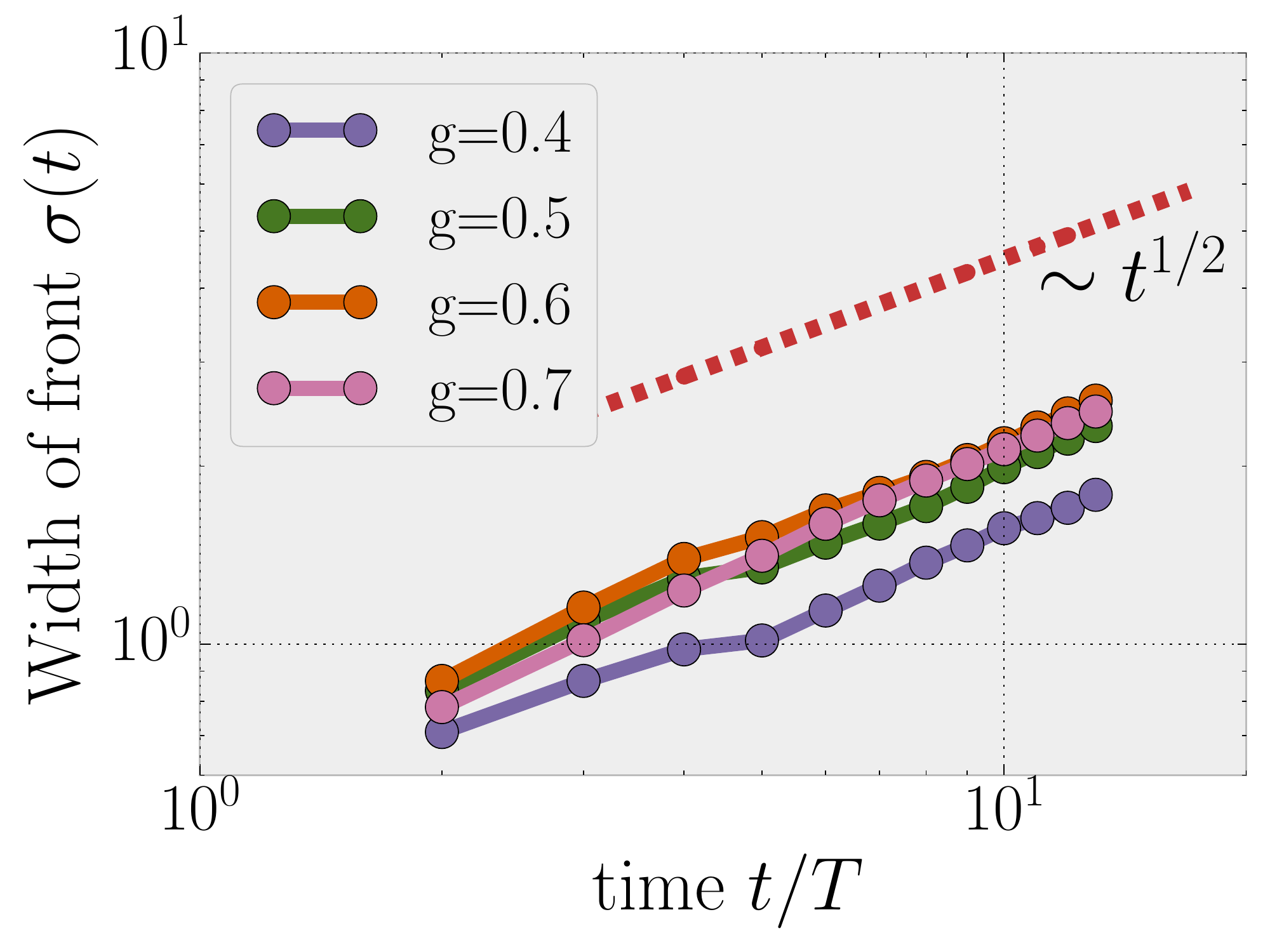} 
\caption{Broadening of the operator wavefront with time in the kicked Ising model for parameters $h=0.809$ and $T = 1.6$. Left: the weight $\rho(s)$ for transverse field $g=0.7$ at times $t/T = 4, 6, \ldots, 14$. Right: the width of the front as defined by Eq.~\eqref{eq:frontwidth} as a function of time for different $g$, showing a roughly diffusive spreading.}
 \label{fig:KimHuse_front_widening}
 \end{figure}

Finally, one might wonder whether the above story continues to hold when we consider a system with energy conservation i.e., a time independent Hamiltonian. We have been informed\cite{Cheryne} that there is evidence of a diffusively growing front in a family of energy conserving ergodic spin-chains. 

 \section{Fractal Clifford circuits}\label{s:Clifford}
In this section we compare the results obtained for the random circuit model of Sec.~\ref{s:randomcircuit} to another set of circuits which do not in general exhibit energy conservation, namely  Clifford circuits. We show that despite the fact that Clifford circuits can be `ergodic' in certain senses -- they can exhibit linear entanglement growth and correlations heating up to infinite temperature -- both their spectrum and their OTOC behavior is anomalous and non-chaotic. 

General Clifford circuits have a particularly simple structure to their operator spreading coefficients: Under time evolution by $t$ steps, a simple Pauli string becomes another Pauli string
\begin{equation}\label{eq:Clifford}
|c_\nu^\mu(t)| = \delta_{\mu,M_t(\mu)}
\end{equation}
where $M_t$ is a linear endomorphism acting on the set of strings $\left(\mathbb{Z}_{q}^{\otimes2}\right)^{\otimes L}$. Thus, a Pauli operator evolves to a single Pauli operator, rather than the superposition of Pauli operators allowed by \eqnref{eq:cmunudef}. $M_t$ has to obey a number of constraints. In particular, time evolution should preserve the commutation relations amongst the Pauli strings\cite{Schlingemann08,Gutschow09,Gutschow10}. (Incidentally, using these constraints, it is possible to classify all  \textit{translation invariant} Clifford circuits into three types called fractal, glider, and periodic\cite{Gutschow09,Gutschow10}.) 
 
In line with the stringency of the constraint \eqnref{eq:Clifford}, it is unsurprising that Clifford circuits have pathological properties distinguishing them from more generic ergodic systems.  Calculating  the OTOC \eqnref{eq:OTOC_def} for initial Pauli strings $\sigma^\mu,\sigma^\nu$ for a Clifford circuit gives
\be
 \frac{1}{2}\langle\psi\mid|[\sigma^{\mu}(t),\sigma^{\nu}]|^{2} \mid \psi \rangle =1 - \cos \theta_{M_t(\mu),\nu}
\ee
where $\theta$ is the phase obtained commuting $\sigma^{M_t(\mu)}$ through $\sigma^{\nu}$, and the final result is independent of the state $\psi$. For generic clifford circuits, this result shows persistent oscillations at late times, and does not settle to a specific limit. This is in contrast to the OTOC decay seen in \secref{ss:randomcircuitOTOC} and the OTOC decay expected in more generic thermalizing systems (see \secref{s:kickedIsing}). 

Additionally, Clifford circuits tend to have pathological spectral properties not associated with ergodic systems. For instance, one can prove that translation invariant Clifford circuits have exact recurrences $U^{t_\text{rec}}_\text{f} \propto1$ on time scales linear in system size $t_\text{rec} = O(L)$ (see \appref{app:linearrecurrence}). This directly implies that the eigenvalues of $U^{t_\text{rec}}_\text{f} $  are $t_{\text{rec}}$-th roots of unity, which in particular implies that the average level degeneracy is $O(2^L/L)$. This spectral structure does not exhibit the level repulsion we expect in systems with ETH and no apparent global symmetries.

Although operators obey the  stringent condition \eqnref{eq:Clifford}, clifford circuits can still exhibit many  ergodic properties usually associated with `ergodicity'. Indeed, it can be proved rigorously that  the above mentioned `fractal' clifford circuits exhibit: (a) Linear entanglement growth starting from certain so-called stabilizer initial states\footnote{Theorem IV.8 in Ref.~\onlinecite{Gutschow10}.} and relatedly (b) starting from an initial product state, at long times all local observables tend towards their infinite temperature expectation values\footnote{Prop III.2 in Ref.~\onlinecite{Gutschow10}.}.  These fractal clifford circuits are periodic in time, in addition to being spatially translation invariant. An explicit example of such a circuit has $q=2$ and takes form $U(t)=U_{\text{f}}^t$ where
\be\label{eq:fractalClifford}
U_\text{f} = e^{\frac{i\pi}{3\sqrt{3}}\sum_s (-X_s+Y_s-Z_s)} e^{\frac{i\pi}{4}\sum_{s} Z_s Z_{s+1}} \text{.}
\ee
Note that the resulting circuit has the geometry illustrated in \figref{fig:circuit} (with the one-site rotations merged into the 2-site gates), and the circuit elements repeat every $2$ layers. Operators evolve under this circuit in a particularly simple way:
\begin{equation}\label{eq:Clifford_op_spread}
Z_s\to Y_s \quad Y_s\to Y_{s-1}X_{s}Y_{s+1}\quad X_{s}\to Y_{s-1}Z_{s}Y_{s+1}.
\end{equation} 
Note that for certain strings there is a possibility of cancelation, e.g. $Y_s Z_{s+1} \to Y_{s-1}X_{s}1\!\!1_{s+1}$.

\figref{fig:Clifford_entanglement} shows that for this particular circuit even more generic initial states show near linear entanglement growth, with a rate that is mostly independent of the initial state at long times. We can explain this using the operator spreading picture of entanglement growth, discussed in~\secref{ss:randomcircuit_entanglement}. Looking at~\eqref{eq:Clifford_op_spread} we notice that a string with a Pauli operator $X$ or $Y$ at its right endpoint will keep moving to the right at a fixed speed of 1 site / 1 period forever, leading to the fixed rate linear entanglement growth seen for stabilizer states. The deviations from this behavior for a random product initial state come from strings that, up to time $t$, failed to start growing to the right. However, such operators need to have a very specific structure and their number is exponentially small in $t$ (going roughly as $4^{-t}$). Consequently, after the first few periods, most product states settle to the same entanglement velocity exhibited by stabilizer states. This is in contrast with the behavior of the \emph{operator entanglement} of any Pauli string $\sigma^\mu(t)$ which remains 0 at all times due to Eq.~\eqref{eq:Clifford}.

 \begin{figure}[h!]
 \centering
  	\includegraphics[width=0.3\textwidth]{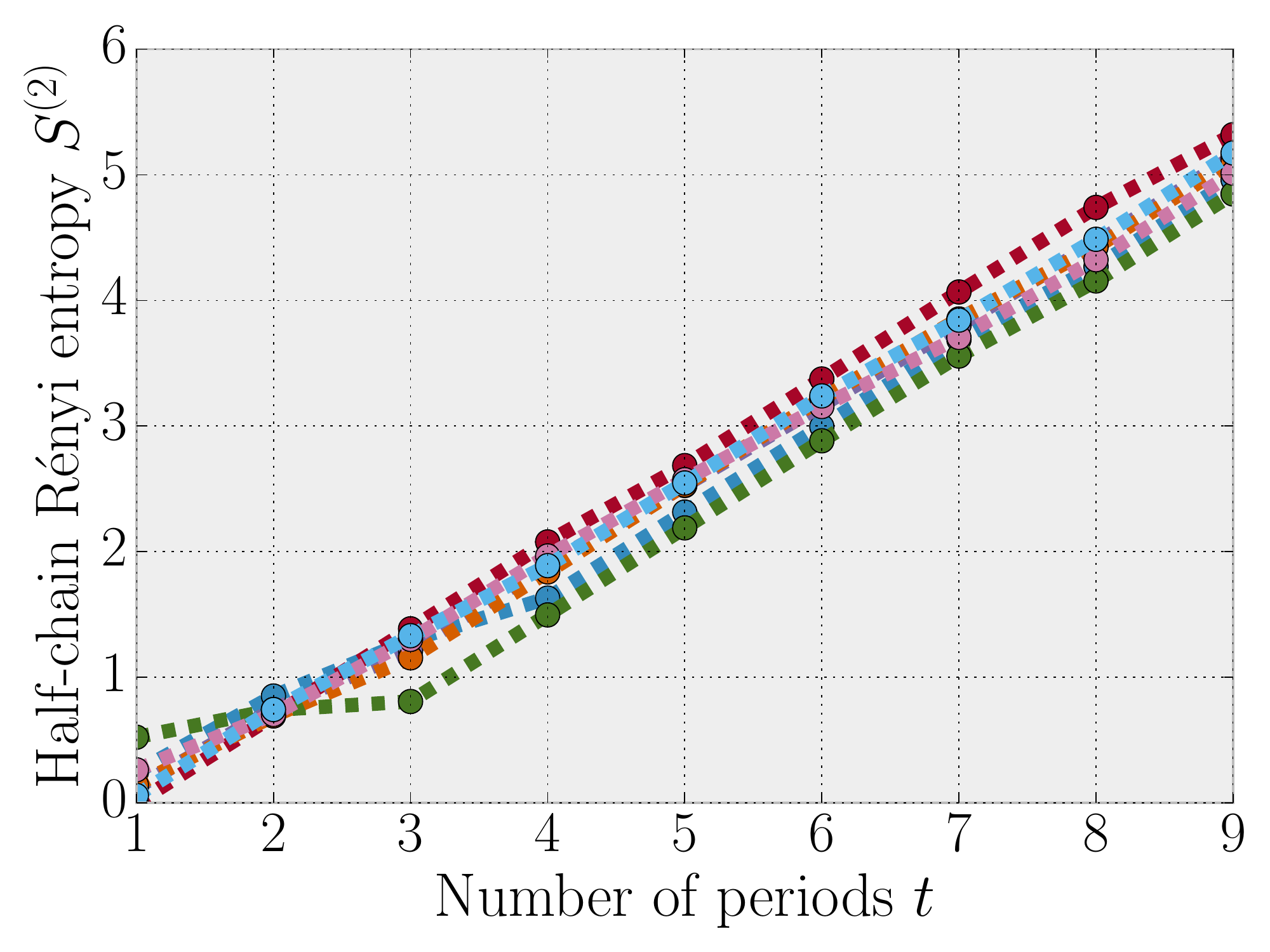} 
\caption{Time evolution of the second R\'enyi entropy between two halves of a 40 site chain in the periodic Clifford circuit defined in Eq.~\eqref{eq:fractalClifford}. The plot shows the growth of entanglement for 7 different initial random product states. After the first few steps the growth rate of entanglement becomes roughly the same for all initial states, set by the fact that the majority of operators travels with a fixed velocity of 1 site / 1 period.}
 \label{fig:Clifford_entanglement}
 \end{figure}

In summary, Clifford circuits can be `ergodic' in the sense that they can exhibit linear entanglement growth and thermalization of local observables, but do not have the spectral or OTOC behavior expected of generic ergodic systems. This suggests that linear entanglement growth is a rather coarse measure of quantum information spreading, sensitive only to the fact that operators tend to grow in extent over time. The OTOC, on the other hand, is sensitive to both the fact that operators grow in extent, and the fact that they become complicated superpositions of many Pauli strings. Note also that the entanglement entropy tests the behavior of a large ensemble of different Pauli strings, while the OTOC characterizes the evolution of a single initial Pauli operator, thus giving more detailed information on the dynamics.
 
\section{Conclusions}\label{s:conclusion}
We considered the spreading of quantum information in one-dimensional systems with local unitary time evolution but no conservation laws. Our key results are as follows. For random circuit systems we derived an exact hydrodynamical description of operator spreading. According to this description, operators grow into superpositions of Pauli strings which tend to be supported at a front, propagating with velocity $v_{\text{B}}$ where $v_{\text{B}}$ is a velocity scale distinct from the light cone velocity (see \figref{fig:front_propagation}). An important consequence of the hydrodynamic equation is that the front itself undergoes a diffusive broadening in time. We proved that the velocity $v_{\text{B}}$ also determines the characteristic scale of change of the out-ot-time-order commutator (OTOC) between two local operators, while at very early times, long before the arrival of the front, we find a regime of exponential growth for the OTOC, with an exponent that depends on the initial separation of the two operators -- it remains to be determined whether this  early time exponential growth represents a quantum analogue of the Lyapunov behavior present in classical chaotic systems. The exact description of operator spreading in this model also allows us to give an exact formula for evolution of (typical) entanglement entropy across a cut, starting from an initial product state. We find that the entanglement grows with a third distinct velocity scale, the entanglement velocity $v_\text{E}$, which is strongly affected by the diffusive behavior, making $v_\text{E} < v_{\text{B}}$.

Comparing our exact results for the random circuit to an ergodic Floquet spin chain, we verified the presence of a diffusively broadening operator front in the kicked Ising model. This leads us to a tentative conjecture that this behavior is present for generic ergodic interacting quantum systems in 1D, evolving under \textit{local unitary} evolution, allowing for a universal coarse grained hydrodynamic description in these systems. We contrasted this with the fine tuned behavior observed in Clifford circuits, for which operators can spread ballistically but do not  become entangled superpositions of many operators. The ballistic spread of operators leads to the linear entanglement growth and the thermalization of local observables seen in certain Clifford circuits, but the fine tuned nature of the operator growth shows up in the pathological behavior of OTOCs. This demonstrates that linear entanglement growth and thermalization are not good predictors of the scrambling behavior captured by the OTOC. 

We note that within existing field theoretic calculations \cite{Stanford2016,Aleiner16,Swingle17} the broadening of the operator wavefront is not present, in contrast with our results. This discrepancy could be due to the unbounded local Hilbert spaces implicit in the field theoretic treatment; in any case, it should be resolved. Another possible direction for future work is finding the exact range of validity of the hydrodynamic operator spreading picture proposed here. It is known for example that the propagation of the OTOC becomes slower than ballistic in strongly disordered (but still thermalizing) systems\cite{Luitz17} and it would it is an interesting question whether this effect can be incorporated into some form of hydrodynamic description.

\paragraph*{Related work:} Shortly before completing this manuscript, we became aware of related work\footnote{A. Nahum, S. Vijay, and J. Haah, Operator spreading in Random Unitary circuits.}, which should appear in the same arxiv posting. We have also been informed of forthcoming related work\footnote{V.~Khemani et al, forthcoming.}.

%\paragraph*{Author Contributions:} Analytical results (CvK and TR), numerics (TR and FP), general theory (all authors).
\acknowledgements
This work was supported by the Princeton Center for Theoretical Science (CvK), Research Unit FOR 1807 through grants no. PO 1370/2-1 (TR and FP), and DOE Grant No DE-SC/0016244 (SLS). We acknowledge useful conversations with Mark Mezei, Douglas Stanford and David Huse. SLS would like to thank Lauren McGough and Roderich Moessner for early discussions about OTOCs.

\appendix
\section{Validity of \eqnref{eq:Rapprox} }\label{app:approx}
It suffices to work in coarse grained co-ordinates $t,x$. Recall that
\be
\overline{R}(x)=\sum_{y\leq x}\overline{\rho_{R}}(y).
\ee
\be
\overline{\rho_{R}}(x,t)=\frac{q^{2(t+x)}}{(1+q^{2})^{2t}}{2t \choose t+x}.
\ee
Our task is to justify the formula  \eqnref{eq:Rapprox}, which can be more carefully phrased as follows:
In the limit, $t\rightarrow\infty$, with $\ensuremath{|x-v_{\text{B}}t|/t=\varepsilon}$
held to be a fixed nonzero number, the integrated operator density
obeys
\begin{align}\label{eq:precisebounds}
\frac{\overline{R}(x)}{\overline{\rho_{R}}(x)} & =c_{0}\text{ for }-1<\kappa<v_{\text{B}} \\
\frac{1-\overline{R}(x)}{\overline{\rho_{R}}(x)} & =c_{1}\text{ for }v_{\text{B}}<\kappa<1
\end{align}
where $c_{0,1}$ are  positive numbers, bounded  in the $t\rightarrow\infty $ limit, and we defined $\kappa\equiv x/ t$ and work in units where $v_{\text{LC}} =1$. Define
\begin{align}
Q_{x} & \equiv\frac{\overline{\rho_{R}}(x+1,t)}{\overline{\rho_{R}}(x,t)}\\
 & =q^{2}\frac{1-\kappa}{1+\kappa+\frac{1}{t}}
\end{align}
The quotient $Q_{x}$ is always positive. It is easy to verify that for $-1<\kappa <v_{\text{B}} $
the quotient is greater than 1, and an increasing function of $x$. On the other hand for $v_{\text{B}}<\kappa<1$, this quotient is less than $1$, and a decreasing function
of $x$.

When $-1<\kappa<v_{\text{B}}$ we use
these facts to bound
\begin{align}
\overline{R}(x) & \leq\overline{\rho_{R}}(x,t)(1+Q_{x}^{-1}+Q_{x}^{-2}\ldots)\nonumber\\
 & =\overline{\rho_{R}}(x,t)\frac{1}{1-Q_{x}^{-1}}
\end{align}
On the other hand, it is immediate that $\overline{R}(x)\geq\overline{\rho_{R}}(x,t)$.
Noting that in the large $t$ limit $Q_{x}=\frac{(1+v_{\text{B}})(1-\kappa)}{(1-v_{\text{B}})(1+\kappa)}$
we find
\begin{align}\label{bound1}
\left|\frac{\overline{R}(x,t)}{\overline{\rho_{R}}(x,t)}-1\right| & \leq\frac{Q_{x}^{-1}}{1-Q_{x}^{-1}}\nonumber\\
 & =\frac{\left(1+\kappa\right)\left(1-v_{\text{B}}\right)}{2\varepsilon}\punc{.}
\end{align}

Hence $c_{0}\leq1+\frac{\left(1+\kappa\right)\left(1-v_{\text{B}}\right)}{2\varepsilon}$is
an $O(1)$ constant.

Similarly we consider $v_{\text{B}}<\kappa<1$. Using $1=\sum_{y}\rho_{R}(y,t)$
it follows that $\overline{R}(x)=1-\sum_{y>x}\rho_{R}(y,t)$. In the
present case $Q_{x}<1$ and is straightforward to derive a similar
bound
\begin{align}\label{bound2}
\left|\frac{1-\overline{R}(x,t)}{\overline{\rho_{R}}(x,t)}-1\right| & \leq\frac{Q_{x}}{1-Q_{x}}=\frac{\left(1-\kappa\right)\left(1+v_{\text{B}}\right)}{2\varepsilon}
\end{align}
Hence $c_{1}\leq1+\frac{\left(1-\kappa\right)\left(1+v_{\text{B}}\right)}{2\varepsilon}$
is an $O(1)$ number in the large $t$ limit. Note that near the edge
of the light cone $\kappa=\mp 1 $ respectively, the results \eqnref{eq:precisebounds}
become increasingly exact as each of the bounds \eqnref{bound1} and \eqnref{bound2} become tighter. On the other hand, as we approach the front $\varepsilon\rightarrow 0$ the bounds become looser and \eqnref{eq:precisebounds} is less reliable -- in this regime, the near front expansion \eqnref{eq:Rdefn} becomes more useful.

\section{Application to OTOC}\label{app:otoc}
Following \secref{ss:randomcircuitOTOC}, the quantity $f(s,\tau)\equiv1-\overline{\mathcal{C}}(s,\tau)$ can
be written as
\be
f(s,\tau)  =  \sum_{\mu}\overline{\left|c^{\mu}(\tau)\right|^{2}}\cos\theta_{\mu,Z_{s}},
\ee
where $e^{i\theta_{\mu,Z_{s}}}$ is a $q^{\text{th}}$ root of unity
acquired by commuting $\sigma^{\mu}$ past $Z_{s}$. We can reparameterize
the sum by summing over the left right endpoints of the Pauli string
$\mu$
\begin{align}
f(s,\tau) & =\sum_{l=-t}^{t-1}\sum_{r=-t}^{t-1}h(l,r)\sum_{\mu\in\mathfrak{F}(l,r)}\cos\theta_{\mu,Z_{s}},\label{eq:1motoc}
\end{align}
where $h(l,r)$is simply the average $\overline{\left|c^{\mu}(\tau)\right|^{2}}$
for a Pauli string $\mu$ with left/right endpoint -- recall that
this value does not depend on the internal structure of $\mu$(see
below eqnref or above ), only on the endpoints of $\mu$. Here $\mathfrak{F}(l,r)\equiv\{\mu:\text{supp}(\mu)=[l,r]\}$.
Those intervals $\left[l,r\right]$ such that $s\in(l,r)$ do not
on net contribute to this sum, because there results a sum over $q^{\text{th}}$
roots of unity disappears
\[
h(l,r)\sum_{\mu\in\mathfrak{F}(l,r)}\cos\theta_{\mu,Z_{s}}=0.
\]
There are also contributions to \eqref{eq:1motoc} which arises when
$s$ is on the left and/or right edge of an interval i.e., $s=l$
or $s=r$.

First, perform the sum over $\mu$ in \eqref{eq:1motoc} for the case $r=s>l$
\begin{align}
&\sum_{\mu\in\mathfrak{F}(l,s)}\cos\theta_{\mu,Z_{s}}\nonumber \\
=& -(q^{-2}\delta_{l<s-1}+\delta_{l=s-1,s}(q^{2}-1)^{-1})\left|\mathfrak{F}(l,s)\right| \punc{.} \label{eq:r>l}
\end{align}
Next, perform the sum over $\mu$ in \eqref{eq:1motoc} for the case $l=s<r$
\begin{align}
&\sum_{\mu\in\mathfrak{F}(s,r)}\cos\theta_{\mu,Z_{s}}  \nonumber\\
  =&-(q^{-2}\delta_{r>s+1}+\delta_{r=s,s+1}(q^{2}-1)^{-1}\left|\mathfrak{F}(s,r)\right|) \punc{.} \label{eq:l<r}
\end{align}

Combining these sum identities \eqnref{eq:r>l}, \eqnref{eq:l<r} back into \eqnref{eq:1motoc} yields
(after some rearrangement)

\begin{align*}
&f(s,\tau)\nonumber \\
 =&\left(\sum_{l\leq r\leq s-1}+\sum_{s+1\leq l\leq r}-\sum_{r=s}\sum_{l\leq s}q^{-2}-\sum_{l=s}\sum_{r\geq s}q^{-2}\right)\times \\
 &\times h(l,r)\left|\mathfrak{F}(l,r)\right|\\
  +&\frac{1}{q^{2}}\left(-(q^{2}-1)(h(s-1,s)+h(s,s+1))+\left(q^{2}-2\right)h(s,s)\right)
\end{align*}

The first line is

\be
f_{1^{\text{st}}}(s,\tau)  =\overline{R}(s-1)+\overline{L}(s+1)-q^{-2}\left(\overline{\rho_{R}}\left(s\right)+\overline{\rho_{L}}\left(s\right)\right).
\ee

It is readily verified that the second line disappears exponentially quickly in $\tau$ as $(2q/(q^2+1))^{2\tau}$, for any
$s$ using  \eqnref{eq:operatorweight}. Hence, in the $\tau\rightarrow\infty$ fixed $s/\tau$ limit,
we can approximate 
\begin{align}
f(s,\tau) & \approx f_{1^{\text{st}}}(s,\tau)\nonumber\\
 & =\overline{R}(s-1,\tau)-q^{-2}\overline{\rho_{R}}\left(s,\tau\right)\\
 & +\overline{R}(-s-2,\tau)-q^{-2}\overline{\rho_{R}}\left(-s-2,\tau\right)
\end{align}
where we used the fact $\overline{L}(s+1)=\overline{R}(-s-2)$ and
$\overline{\rho_{L}}(s+1)=\overline{\rho_{R}}(-s-2)$. For $0<s<\tau/v_{\text{LC}}$,
and for any $q$, we can use \eqnref{eq:Rdefn} and
\eqnref{eq:rho_realspace} to argue that the second
line is suppressed by a factor of $\sim q^{-s}$ relative to first
line. Therefore, provided $\kappa>0$ in the $\tau\rightarrow\infty$
limit, the OTOC behaves as

\be
f(s,\tau)\approx\overline{R}(s-1,\tau)-q^{-2}\overline{\rho_{R}}\left(s,\tau\right)
\ee
up to exponentially small corrections in $\tau$.

\section{Exact operator spread coefficients}\label{App:cmunu}
In what follows we give an exact expression for the averaged operator spread coefficients, and a sketch of the derivation. We leave a more detailed derivation to future work. The averaged operator spread coefficients are defined as
\begin{align}
\overline{|c_{\vec{\mu}}^{\vec{\nu}}\left(t\right)|^{2}} & =\int\left[\prod_{\tau=1}^{t,\leftarrow}\prod_{j=1}^{L}d_{\text{Haar}}W\left(2j+p_{\tau},\tau\right)\right]\nonumber \\
 & \times\left|\frac{1}{q^{M}}\text{tr}\left(\sigma^{\vec{\nu}\dagger}U^{\dagger}(t)\sigma^{\vec{\mu}}U(t)\right)\right|^{2}\label{eq:Haaraveragec}\punc{.}
\end{align}
The Haar-averaging can be performed explicitly using the identity~\eqref{eq:Haaraverage}. After averaging, each two-site gate can be represented by a classical, Ising-like variable taking only two possible values. Due to the geometry of the circuit, these Ising variables form a triangular lattice. Eq.~\eqref{eq:Haaraveragec} becomes a classical partition function, i.e. a sum over all possible spin configurations. The partition function consists of two edge parts, which depend on the configurations on the first (last) layer and the Pauli strings $\vec{\nu}$ ($\vec{\mu}$), and a bulk part which is independent of the Pauli strings in question. Due to the Haar-averaging, the only information that remains in the partition function about the strings $\vec{\mu}$ and $\vec{\nu}$ is which sites they act on non-trivially. The bulk transfer matrix enforces a light cone structure on the spin variables. A light cone with velocity $v_\text{LC} = 1$ emanates from the two endpoints of the string $\vec{\nu}$ such that all spins outside of the light cone have to point up (otherwise the configuration has zero weight in the partition function for $\overline{|c_{\vec{\mu}}^{\vec{\nu}}(t)|^{2}}$).

In the case of an initial Pauli operator acting on a single site the partition function for the operator spreading can be evaluated exactly. In this case the fact that $\vec{\mu}$ acts on one site only yields a boundary condition for the partition function wherein in the first row there is a single spin pointing down while all others point up. The partition function then becomes a sum over all possible ways this initially one-site domain can spread within the light cone, as shown in Fig.~\ref{fig:domain_spread}. Furthermore, the bulk interaction terms are only non-trivial at the boundary between the two domains and consequently give the same contribution for all domain configurations with the same depth. Thus the calculation simplifies to counting the possible domain configurations which can be done by considering it as a two-particle random walk for the two endpoints of the domain.

 \begin{figure}[h!]
 \centering
  	\includegraphics[width=0.2\textwidth]{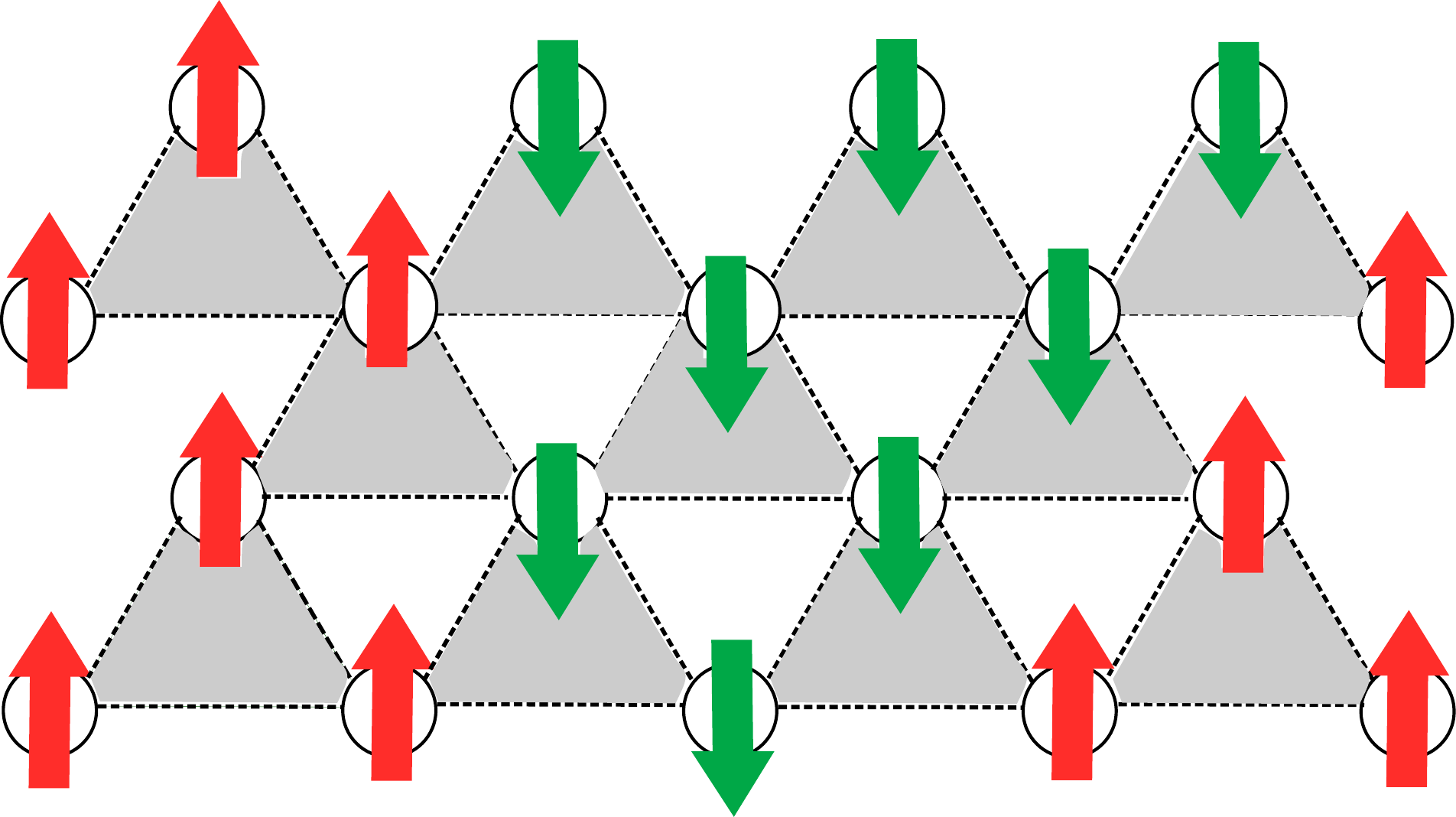} 
\caption{Example of a classical spin configuration contributing to the operator spreading coefficients of an initial one-site Pauli operator. All such configurations have a single domain of down spins spreading inside the light cone and each of these configurations contributes equally to $\overline{|c^{\vec{\nu}}(t)|^{2}}$.}
 \label{fig:domain_spread}
 \end{figure}

The calculation outlined above yields all average squared coefficients $\overline{|c_{\vec{\mu}}^{\vec{\nu}}(t)|^{2}}$ where $\mu_j = 0$ if $j\neq 0$. In the following we simplify the notation by dropping the first index and denoting these as $\overline{|c^{\vec{\nu}}(t)|^{2}}$. The exact formula for these squared coefficients is given in Eq.~\eqref{eq:Cmunufull}. A surprising property of this formula only depends on the positions of the right and left endpoints $l,r$ of the Pauli string $\vec{\nu}$, and not on more detailed information concerning the internal structure of the string i.e., $\overline{|c^{\vec{\nu}}(t)|^{2}}=h(l,r)$. This is a consequence of the Haar averaging which in each step washes out all memory of the internal structure.

Form a circuit with even number $D$ layers. Number the spaces between
the two site unitaries in the last layer $-D/2,1-D/2,\ldots D/2$.
The support of an operator string $\vec{\nu}$ can be represented
by $x,y\in\{ -D/2,1-D/2,\ldots+D/2\} $ with $x<y$. The
average square coefficient is obtained by plugging $x,y$ into the
formula.

\begin{align}
\overline{|c^{\vec{\nu}}|^{2}}  =& \frac{1}{(1+q^{-2})^{2(D-1)}} \frac{q^{-2(y-x)-2D}}{1-q^{-4}} \mathcal{J}(x,y,D)\nonumber \\
\mathcal{J}(x,y,D) \equiv& \sum_{0\leq a\leq\frac{D}{2}+x} \sum_{0\leq b\leq\frac{D}{2}-y}\binom{D}{\frac{D}{2}-b-y}\nonumber\\
&\times \binom{D}{\frac{D}{2}-a+x} q^{-2a-2b}\frac{b+y+a-x}{D}
\label{eq:Cmunufull}
\end{align}

Note that this expression depends only on $x,y$ and $q$ so we also
denote it as $\overline{|c_{\vec{\nu}}^{\mu_{1}}|^{2}}=h(x,y,D)$
where we drop the $q$ dependence for simplicity. These expressions
are of-course complicated. We note as an aside that this formula has
a slightly neater expression in terms of hypergeometric functions.

\subsection{Useful limits}
Let us calculate the weight on an operator with endpoints $x,y$, in the large $D$ limit (in circuit co-ordinates) starting from \eqnref{eq:Cmunufull}. Re-express the weight as
\begin{align*}
\overline{\left|c_{\vec{\nu}}^{\mu_{1}}\right|^{2}} & =q^{-2D}\times\frac{1}{\left(1+q^{-2}\right)^{2(D-1)}}\frac{1}{1-q^{-4}}\mathcal{G}\left(x,y,D\right)\\
\mathcal{G}\left(x,y,D\right) & \equiv\sum_{0\leq a\leq\frac{D}{2}+x}\times\sum_{0\leq b\leq\frac{D}{2}-y}q^{-2b-2y}q^{-2a+2x}\\
&\times\binom{D}{\frac{D}{2}-b-y}\binom{D}{\frac{D}{2}-a+x}\frac{b+y+a-x}{D}
\end{align*}

Defining:
\begin{align*}
\mathcal{H}(x,y,D) \equiv q^{-2D}(1+q^{2})^{2D}v_{D}(\frac{D}{2}+x)v_{D}(\frac{D}{2}-y)\\
\end{align*}
where 
\be
v_{D}(\Delta) \equiv (1+q^{2})^{-D}\sum_{j=0}^{\Delta}{D \choose j}q^{2j}
\ee
we find that
\be
\frac{d\mathcal{H}(x,y,D)}{dq} =-2Dq^{-1}\times\mathcal{G}
\ee
Putting this altogether, we have
\begin{align}
h(x,y,D) & =-\frac{1}{2D}\frac{q^{-2D}}{\left(1+q^{-2}\right)^{2(D-1)}}\frac{1}{1-q^{-4}} \nonumber\\
&\times\frac{d}{d\log q}[q^{-2D}(1+q^{2})^{2D}v_{D}(\frac{D}{2}+x)v_{D}(\frac{D}{2}-y)] \label{hfunct}
\end{align}
For fixed $x,y$ taking the large $D$ limit, the $v_D$ functions can be approximated as (using the same reasoning as in \eqnref{eq:Rapprox} ):
\begin{align*}
v_{D}(\frac{D}{2}+x) & \sim\sqrt{\frac{2}{\pi D}}\frac{q^{2(x+1)}\left(\frac{2q}{q^{2}+1}\right)^{D}}{q^{2}-1}
\end{align*}
Plugging in this approximation we find.
\begin{align}
h(x,y)  &\sim \left(\frac{2q}{q^{2}+1}\right)^{2D}\nonumber\\
&\times\frac{(q^{2}+1)((q^{2}-1)(y-x)+2)q^{2(x-y+2)}}{\pi D^{2}\left(q^{2}-1\right)^{4}} \label{eq:operatorweight}
\end{align}
So, fixing $x,y$, the weight on an operator with endpoints $x,y$ decays exponentially quickly $~\left(\frac{2q}{q^{2}+1}\right)^{2D}$ with large $D$. 

\section{Long time correlations}\label{ss:longtimerandom}
In \secref{s:intro} we anticipated that random unitary circuits should `heat to infinite temperature', much like ergodic Floquet circuits. Here we back up this claim by examining the long time behaviour of various correlations functions, demonstrating that they relax to their expected infinite temperature values. For simplicity we consider 1 and 2 point functions for specific 1-site Pauli operators of form $\sigma_0^\alpha$ -- although most of the results below follow for more general operators as well. Fix any initial state $\omega$ and times $t_2 , t_1$. We find that
\be
\overline{\langle{\sigma_0^\alpha(t_2) \sigma_z^\beta(t_1)}\rangle_\omega} =0\punc{.}
\ee
This follows from two observations. First, the operator in the expectation value can be written $U_{01} U_{12}\sigma_0^\alpha U_{21} \sigma_x^\beta U_{10}$ where $U_{ij} = U^{-1}_{ji}$ is shorthand for the unitary evolutions between times $t_i$ and $t_j$. Now $U_{12}$ is statistically independent from $U_{01}$, so we can average over these disjoint circuits independently. Provided $t_1 \neq t_2$ and $\alpha\neq 0$ it is straightforward to see that the Haar average $\overline{U_{12}\sigma_0^\alpha U_{21}}=0$. This can be re-expressed succinctly as $\overline{c^{\alpha}_{\mu}(t)}=0$ for all $t= t_2 - t_1\neq 0,\alpha \neq 0$. Note that this result is independent of the initial state $\omega$ and the value $\beta$ -- in particular we can recover the behavior of 1 point functions by setting $\beta=0$.

While the above correlation functions disappear on average, we may also quantify how their variance behaves at long times. Indeed, the variance decays exponentially in time, at least for random initial product states $\psi$. The variance is  
\be
\overline{\langle{\sigma_0^\alpha(t_2) \sigma_z^\beta(t_1)}\rangle^2_\omega} = \sum_{\mu} \overline{|c^{\alpha}_{\mu}(t)|^2}  |\langle{\sigma^\mu(t_1)\sigma_z^\beta(t_1)}\rangle_\omega|^2\punc{.}
\ee
We argue that this variance vanishes as we increase $t\rightarrow\infty$ while fixing $t_1$. We can show this rigorously for an infinite temperature state in \appref{app:infinitetemperature}.

For a random product state we have a less rigorous argument which proceeds as follows. First, given any $\epsilon>0$, it will be true that for sufficiently long times $t$, all of the weight of $\alpha$ is invested in strings $\mu$ with left/right endpoints $l < -(v_{\text{B}}-\epsilon)t, r>(v_{\text{B}}-\epsilon)t$ respectively, up to exponentially small corrections in $t$. These statements follow from \eqnref{eq:rhonearfront}. Second, as $\alpha$ is a 1-site operator, $ \overline{|c^{\alpha}_{\mu}(t)|^2}$ only depends on the endpoints of $\mu$ rather than the detailed internal structure (See the discussion under Eq.~\eqref{eq:OTOC_from_c}). Hence, up to exponentially small corrections in time, $\sigma_0^\alpha(t)$ is made up of an equal amplitude superposition of all operators $\mu$ with left/right endpoints near $\mp v_{\text{B}} t$ respectively. The vast majority of such strings contain an extensive $O(v_{\text{B}} t)$ number of Pauli operators. The  expectation values of such strings on a random product states is exponentially decaying in the number of Pauli operators. As a typical $\mu$ string contains $O(v_{\text{B}} t)$ Pauli operators, we find $ |\langle{\sigma^\mu\sigma_z^\beta}\rangle_\omega|^2 \sim e^{- \zeta v_{\text{B}} t }$ for some constant $\zeta$. Hence for $t_1$=0, $\overline{\langle{\sigma_0^\alpha(t_2) \sigma_z^\beta(0)}\rangle^2_\psi} \sim e^{-\zeta v_{\text{B}} t}$.  For $t_1$ nonzero, we expect even more marked decay $\overline{\langle{\sigma_0^\alpha(t_2) \sigma_z^\beta(t_1)}\rangle^2_\psi} \sim e^{-\zeta v_{\text{B}} (t+t_1)}$ because the support of $\sigma^\mu\sigma_z^\beta$ is further increased under time evolution.

\subsection{Infinite temperature results}\label{app:infinitetemperature}
Consider the variance of the infinite temperature expectation value function 
\begin{align}
\overline{\tr (2^{-L}(\sigma_0^\alpha(t_2) \sigma_z^\beta(t_1))^2} &= \sum_{\mu} \overline{|c^{\alpha}_{\mu}(t)|^2}  |2^{-L} \tr (\sigma^\mu\sigma_z^\beta) |^2\nonumber\\
&= \overline{|c^{\alpha}_{\beta_z}(t)|^2} \punc{.}
\end{align}
This is a Haar averaged single site weight. In the large $t$,  $z/t\rightarrow 0$ limit, we have an expression for this quantity (in the coarse grained lattice basis). It is approximately equal to \eqnref{eq:operatorweight} using $x= z/2, y= x+1,D=2 t$. Asymptotically then the variance in the infinite temperature average decays exponentially as
\be
\sim \left(\frac{2q}{q^{2}+1}\right)^{4t}\punc{.}
\ee

\section{Haar Identities}
Behold the following Haar moments for $d\times d$ random unitary
matrices. 
\[
\int d_{H}U\times U_{i_{1}i_{2}}U_{\overline{i}_{1}\overline{i}_{2}}^{*}=\frac{1}{d}\delta_{i_{1}\overline{i}_{1}}\delta_{i_{2}\overline{j}_{2}}
\]
The following higher moment is useful for deriving the average square operator spread coefficients. 
\begin{align*}
& \int d_{H}U\times U_{i_{1}^{1}i_{2}^{1}}U_{\overline{i}_{1}^{1}\overline{i}_{2}^{1}}^{*}U_{i_{1}^{2}i_{2}^{2}}U_{\overline{i}_{1}^{2}\overline{i}_{2}^{2}}^{*} \\
& =\frac{1}{d^{2}-1}\left(\delta_{i_{1}^{1}\overline{i}_{1}^{1}}\delta_{i_{1}^{2}\overline{i}_{1}^{2}}\times\delta_{i_{2}^{1}\overline{i}_{2}^{1}}\delta_{i_{2}^{2}\overline{i}_{2}^{2}}+\delta_{i_{1}^{1}\overline{i}_{1}^{2}}\delta_{i_{1}^{2}\overline{i}_{1}^{1}}\times\delta_{i_{2}^{1}\overline{i}_{2}^{2}}\delta_{i_{2}^{2}\overline{i}_{2}^{1}}\right)\\
 & -\frac{1}{d(d^{2}-1)}\left(\delta_{i_{1}^{1}\overline{i}_{1}^{2}}\delta_{i_{1}^{2}\overline{i}_{1}^{1}}\times\delta_{i_{2}^{1}\overline{i}_{2}^{1}}\delta_{i_{2}^{2}\overline{i}_{2}^{2}}+\delta_{i_{1}^{1}\overline{i}_{1}^{1}}\delta_{i_{1}^{2}\overline{i}_{1}^{2}}\times\delta_{i_{2}^{1}\overline{i}_{2}^{2}}\delta_{i_{2}^{2}\overline{i}_{2}^{1}}\right)
\end{align*}

This can be more elegantly expressed as a sum over elements of permutation
group $S_{2}$

\begin{eqnarray}
\int d_{H}U\times U_{i_{1}^{1}i_{2}^{1}}U_{\overline{i}_{1}^{1}\overline{i}_{2}^{1}}^{*}U_{i_{1}^{2}i_{2}^{2}}U_{\overline{i}_{1}^{2}\overline{i}_{2}^{2}}^{*} & = & \sum_{\sigma,\eta\in S_{2}}\text{Wg}(\eta)R(\eta\sigma)_{\overline{i}_{1}^{1}\overline{i}_{1}^{2}}^{i_{1}^{1}i_{1}^{2}}R(\sigma)_{\overline{i}_{2}^{1}\overline{i}_{2}^{2}}^{i_{2}^{1}i_{2}^{2}}\nonumber\\
\text{where, }\text{Wg}_{d}(\eta) & \equiv & \frac{1}{d^{2}-1}\left(\frac{-1}{d}\right)^{\delta_{\eta=(1,2)}}\label{eq:Haaraverage}
\end{eqnarray}

\section{Recurrence times in Translation invariant Clifford quantum cellular
automata (CQCA)}\label{app:linearrecurrence}
The aim of this section is to show that translation invariant Clifford circuits have linear in system size recurrence times (for a certain family of system sizes). We utilize the technology and formalism of Ref.~\onlinecite{Gutschow09}.
\begin{fact}
(Ref.~\onlinecite{Gutschow09} Theorem II.5) Clifford quantum circuits with translation
symmetry (unit cell size 1) are in correspondence with the set of
$2\times2$ matrices $t$ with elements in $\mathbb{Z}/2\mathbb{Z}\left[u,u^{-1}\right]$
(polynomials in $u,u^{-1}$ over ring $\mathbb{Z}/2\mathbb{Z}$) obeying
\begin{align*}
\det t & =u^{2a}\\
t_{ij}(u) & =u^{a}\times\text{Symmetric Laurent}[u,u^{-1}]\\
t_{11}\text{coprime }t_{21}, & t_{12}\text{coprime }t_{22}\\
\end{align*}
where by coprime, we mean the polynomials over ring $\mathbb{Z}/2\mathbb{Z}$
do not possess any common non-trivial factors. We say a clifford quantum
circuit is centered if $a=0$. \label{fact:cliffordclassification}
\end{fact}
\begin{fact}
From Ref.~\onlinecite{Gutschow09} Proposition II.13). A centered CSCA $t$ is periodic
with period c+2 if $\text{tr }t=c$ for$c\in\mathbb{Z}/2\mathbb{Z}$
(so the period is either $2$ or $3$). \label{fact:periodiccircuits}\end{fact}
\begin{proof}
This is actually quite straightforward. From Fact \ref{fact:cliffordclassification},
we have $t^{2}=\text{tr}\left(t\right)t+1$. If $\text{tr}(t)=0$
then $t^{2}=1$. Else $t^{3}=t(t+1)=t+1+t=1$. \end{proof}
\begin{theorem}
Translation invariant centered CQCAs $U$ have linear in system size
recurrence times (at most $t_{\text{rec}}=12L$), at least for system
sizes $L=2^{n}$.\end{theorem}
\begin{proof}
Let $t$ be the $2\times2$ matrix corresponding to clifford unitary
$U$. By the Cayley Hamilton (CH) theorem (which holds for matrices
over arbitrary rings), and the fact that $\det t=1$ we have

\[
t^{2}=\text{tr}\left(t\right)t+1
\]

where the matrices are written over the ring of Laurent polynomials
with $\mathbb{Z}/2\mathbb{Z}$ coefficients i.e., $\mathbb{Z}/2\mathbb{Z}\left[u,u^{-1}\right]$.
Recall from Fact \ref{fact:cliffordclassification} that
for a centered circuit, the elements of $t$ are \emph{symmetric}
Laurent polynomials. Note that the CH theorem can be iterated

\begin{eqnarray*}
t^{4} & = & \left(\text{tr}\left(t\right)t+1\right)^{2}\\
 & = & t^{2}\text{tr}\left(t\right)^{2}+1+2\times(\text{stuff})\\
 & = & t^{2}\text{tr}\left(t\right)^{2}+1\\
 & = & t\text{tr}\left(t\right)^{3}+\text{tr}\left(t\right)^{2}+1\\
 & = & tp^{3}+p^{2}+1
\end{eqnarray*}

where for convenience we denote $p\equiv\text{tr}\left(t\right)$,
which is of course a symmetric Laurent polynomial. Proceeding inductively,
and squaring successive equations, one can show

\begin{align*}
t^{2^{k}}= & t\times p{}^{2^{k}-1}+1\sum_{a=1}^{k}p^{2^{k}-2^{a}}
\end{align*}

Now consider the trace of $t^{2^{k}}$:

\begin{eqnarray*}
\text{tr}\left(t^{2^{k}}\right) & = & p{}^{2^{k}}+2\times\sum_{a=1}^{k}p^{2^{k}-2^{a}}\\
 & = & p^{2^{k}}
\end{eqnarray*}

Now we can write

\[
p=\sum_{r}c_{r}u^{r}
\]

where $c_{r}=c_{-r}\in\mathbb{Z}/2\mathbb{Z}$. Now iteratively square
this expression. The cross terms disappear because the ring is $\mathbb{Z}/2\mathbb{Z}$

\begin{eqnarray*}
p^{2} & = & \sum_{r}c_{r}u^{2r}\\
p^{4} & = & \sum_{r}c_{r}u^{4r}\\
 & \ldots\\
p^{2^{k}} & = & \sum_{r}c_{r}u^{2^{k}r}
\end{eqnarray*}

For a system of size $L$, and periodic boundary conditions, the constraint
$u^{L}=u^{0}=1$ is imposed on our polynomial ring. Hence, setting
$L=2^{n}$ and $k=n$ we get

\[
\text{tr}\left(t^{2^{n}}\right)=p^{2^{n}}=\sum_{r}c_{r}=c_{0}
\]

But $c_{0}$ is just a constant. Using Fact \ref{fact:periodiccircuits},
and the fact $\text{tr}t^{2^{n}}$ a constant, we have $U^{L=2^{n}}$
is a periodic circuit. By this, G{\"u}tschow \textit{et al}. mean that $t^{2^{n}(c_{0}+2)}=\text{1}$ is
the identity matrix. This implies that $U^{2^{n}(c_{0}+2)}$ does
not permute Pauli matrices -- it only multiplies them by phases (which
can only be $\pm1$ in order to preserve generating relations for
Pauli matrices). This in turn implies that $U^{2^{n+1}(c_{0}+2)}\propto1$,
the identity matrix on the many-body Hilbert space. Hence, as $c_{0}+2=2,3$
both divide $6$, we certainly have $U^{2^{n+1}\times6}\propto1$.
In other words $U^{12L}\propto1$ for $L=2^{n}$, and any clifford
circuit obeying the conditions of this theorem. The upshot is that
all such clifford circuits obeying the conditions of the theorem have
a linear in system size recurrence time $t_{\text{rec}}\leq 12 L$ , for $L=2^{n}$.\end{proof}
\begin{corollary}
Translation invariant CQCAs have, on average, and exponential in system
size level degeneracy. \end{corollary}
\begin{proof}
Straightforward. There are $2^{L}$ states in the Hilbert space. The
linear in system size $\kappa L$ recurrence time (e.g., $\kappa=12$) means the eigenvalues
are WLOG $\kappa L$-th roots of unity. Hence, the unitary has at
most $\kappa L$ eigenvalues. Hence, the average level degeneracy
is $2^{L}/\kappa L$ -- exponentially large in system size. 
\end{proof}

\end{document}